\documentclass[journal,12pt,onecolumn,draftcls,]{IEEEtran}

\usepackage[square,sort,compress,comma,numbers]{natbib}

\usepackage{xspace}
\usepackage{filecontents}
\usepackage[draft=true]{hyperref}

\usepackage{balance}

\usepackage{float}
\usepackage{pgf}
\usepackage{tikz}
\usepackage{graphicx}

\usepackage[cmex10]{amsmath}

\usepackage{bm}
\usepackage{soul,color} %
\usepackage{amssymb} 

\usepackage{pifont}
\newcommand{\cmark}{\ding{51}}%
\newcommand{\xmark}{\ding{55}}%

\usepackage{array}

\usepackage{subfigure}
\usepackage{caption2} 

\usepackage{amsfonts}
\usepackage{amsthm}   %

\usepackage{enumerate} %

\usepackage{extarrows}

\usepackage{tkz-fct}

\newtheorem{theorem}{Theorem} 
\newtheorem{lemma}{Lemma}

\newtheorem{proposition}{Proposition}
\newtheorem{corollary}{Corollary}
\newtheorem{remark}{Remark}
\newtheorem{definition}{Definition}

\usetikzlibrary{patterns,snakes}

\usepackage{enumerate}

\usepackage{extarrows}

\newcommand{\myexpectation}[1]{\mathsf{E}{\left[#1\right]}}

\newcommand{\mse}{\mathsf{MSE}}

\newcommand{\pw}{\mathsf{PW}}

\usepackage{url}

\usepackage{mathtools}

\mathtoolsset{showonlyrefs=true}

\begin{document}

\title{Over-the-Air Computation Systems: Optimization, Analysis and Scaling Laws}
\author{Wanchun Liu, \emph{Member, IEEE}, Xin Zang, Yonghui Li, \emph{Fellow, IEEE},\\ Branka Vucetic, \emph{Fellow, IEEE}\\
}
\maketitle
\vspace{-2cm}
\begin{abstract}
	\let\thefootnote\relax\footnote{
		The authors are with School of Electrical and Information Engineering, The University of 
		Sydney, Australia. Emails:	\{wanchun.liu,\ xin.zang,\ yonghui.li,\ branka.vucetic\}@sydney.edu.au.	
	}
For future Internet-of-Things based Big Data applications, data collection from ubiquitous smart sensors with limited spectrum bandwidth is very challenging.
On the other hand, 
to interpret the meaning behind the collected data, 
it is also challenging for an edge fusion center running computing tasks over large data sets with a limited computation capacity.
To tackle these challenges, by exploiting the superposition property of  multiple-access channel and  the functional decomposition, the recently proposed technique, over-the-air computation (AirComp), enables an effective joint data collection and computation from concurrent sensor transmissions.
In this paper, we focus on a single-antenna AirComp system consisting of $K$ sensors and one receiver.
We consider an optimization problem to minimize the computation mean-squared error (MSE) of the $K$ sensors' signals at the receiver by optimizing the transmitting-receiving (Tx-Rx) policy, under the peak power constraint of each sensor.
Although the problem is not convex, we derive the computation-optimal policy in closed form.
Also, we comprehensively investigate the ergodic performance of the AirComp system,
and the scaling laws of the average computation MSE (ACM) and the average power consumption (APC) of different Tx-Rx policies with respect to $K$.
%
For the computation-optimal policy, we 
show that the policy has a vanishing ACM and a vanishing APC with the increasing $K$.




\end{abstract}
\vspace{-0.7cm}
\vspace{-0.8cm}
\section{Introduction}\label{sec:intro}
Under the fast development of wireless communication, networking, data collection and storage, the era of Big Data has arrived~\cite{wu2013data}.
According to a recent DOMO technical report~\cite{domo}, more than $2$ quintillion bytes of data are created every day, and about $1.7$ megabytes of data will be created every second per single person on earth by 2020.
Also, the Internet of Things (IoT), which connects smart devices that interact with each other and collect all kinds of data, is exponentially growing from $2$ billion devices in 2006 to a predicted $200$ billion by 2020, and will be one of the primary drivers of data explosion.
How to effectively collect and leverage Big Data and interpret the meaning behind it have attracted much attention in the areas of public health, manufacturing, agriculture and farming, energy, transportation, supply chain management and logistics.
In such IoT-based Big Data applications (e.g., smart cities), wireless data collection from ubiquitous massive smart sensors/devices with limited spectrum bandwidth is very challenging, especially when the data needs to be dealt in a timely manner. 
On the other hand, 
due to a large number of data sources, we do not care much about the value of each individual data source anymore, but shift our focus on the fusion of massive data and unleash its power, which is actually a computing problem. The computation of a large amount of data is also challenging for an edge devices with a limited computation capacity.


To tackle these challenges, the technique, over-the-air computation (AirComp), has been developed to enable an efficient data-fusion of sensing data from many concurrent sensor transmissions by leveraging the inherent broadcast nature of wireless communications and the application of a beautiful mathematical tool of function representation.

%

\subsection{What is AirComp?}
\subsubsection{Preliminaries}
Assume an ideal multiple-access channel (MAC) of $K$ sensors with the signal-superposition property that 
\begin{equation}\label{idealMac}
r = \sum_{k=1}^{K} u_k,
\end{equation}
where $u_k$ is the transmitted signal of sensor $k$ and $r$ is the received signal at the receiver.

Consider $K$ wireless smart sensors, each having a measurement signal $s_k \in \mathbb{R}, \forall k \in \{1,\cdots,K\}$, and a $K$-variate computing task (e.g., sum, multiplication and arithmetic/geometric mean) $\phi:\mathbb{R}^K \rightarrow \mathbb{R}$ at a designated receiver.
By using a mathematical property in the area of theoretical computer science that 
every real-valued multivariate function is representable in its nomographic form as a function of a finite sum of univariate functions~\cite{buck1979approximate}, there always exists a set of pre-processing functions $\psi_k:\mathbb{R}\rightarrow \mathbb{R},\forall k\in \{1,\cdots,K\}$ and a post-processing function $\varphi:\mathbb{R}\rightarrow \mathbb{R}$ such that
\begin{equation}\label{aircomp}
\phi(s_1,\cdots,s_K)=\varphi\left(\sum_{k=1}^{K}\psi_k(s_k)\right).
\end{equation}

\subsubsection{The Brief Idea} 
Based on \eqref{idealMac} and \eqref{aircomp}, the overall idea of AirComp proposed in~\cite{GoldenbaumTCOM,GoldenbaumTSP} is to let each sensor pre-process its own signal and send $u_k=\psi_k(s_k)$ to the receiver simultaneously, and the receiver processes the received sum of signals $\sum_{k=1}^{K}\psi_k(s_k)$ with function $\varphi(\cdot)$ and thus obtains the desired computation of $K$ sensors' measurement signals $\phi(s_1,\cdots,s_K)$.
Therefore, the transmission and computation of a large number of sensors' signals of an AirComp system are completed in one single time slot (i.e., in a symbol level) in contrast to an intuitive one-by-one-transmit-then-compute protocol.
In other words, AirComp effectively integrates communication and computation by harnessing interference for computing~\cite{GoldenbaumTSP}.
Moreover, the receiver's original computation task $\phi(\cdot)$ of processing $K$ signals has been decomposed into $(K+1)$ small tasks $\{\psi_1(\cdot),\cdots,\psi_K(\cdot),\varphi(\cdot)\}$, and each sensor or the receiver only needs to take one lightweight task with only one signal to be processed.
In this way, the computation complexity of the receiver is significantly reduced, especially when $K$ is large.

\subsubsection{Use Cases} 
In addition to smart-city applications, e.g., using unmanned aerial vehicles for real-time data computation and collection from sensors embedded in ground vehicles, buildings and other infrastructures,
AirComp has been developed and applied in two important emerging mobile applications with highly integrated communication and computation tasks:

	$\bullet$ Over-the-air consensus. For example, in drone swarming and connected car platooning applications, each node concurrently sends its current states including velocity and acceleration in real-time, while the central controller can receive and compute the average current states of the mobile nodes, and generate control commands to drive each node approaching a consensus status~\cite{consensus}.
	
	$\bullet$ Wireless distributed machine learning. In wireless machine learning applications that adopt distributed stochastic gradient descent algorithms for model training, each mobile user calculates the gradient of its local cost function of its own data set in terms of the model parameters and concurrently sends it to the parameter server (i.e., a base station), which then computes a weighted average of the gradients and broadcasts it to the mobile users for further iteration until convergence~\cite{GuangxuWangyong,Gunduz,Osvaldo}.

%
%
%
%

\subsection{Previous Work}
The research of AirComp mostly focuses on two aspects: the pre-processing and post-processing functions design in \eqref{aircomp}, and the analysis of the impact of practical wireless MAC (rather than the ideal one in \eqref{idealMac}) on the performance of AirComp and the transceiver design issues for reducing the impact.
For the former, the theoretical properties and how to design the pre- and post-processing functions with a given multi-variable target function $\phi(\cdot)$ (e.g. geometric mean) have been extensively investigated in~\cite{GoldenbaumTCOM,GoldenbaumTSP,GoldenbaumWIOPT,Katabi}.
For the latter, the computation of the sum of pre-processed signals in \eqref{idealMac}, $\sum_{k=1}^{K}u_k$, is not perfect due to  the non-zero receiver noise and unequal channel coefficients, and hence a key design target is to make the computation distortion of the sum of signals as small as possible.
For a multi-antenna AirComp system, an optimization problem of transmitting and receiving beamforming design was considered to minimize the computation distortion~\cite{GuangxuMIMO}, based on which a wireless-powered AirComp system was studied in \cite{Xiaoyang}.
%
%
Also, the effect of the lack of synchronization between different sensors and the imperfect channel estimation on the distortion of the sum of signals were studied in \cite{GoldenbaumTCOM} and \cite{Goldenbaumletter}, respectively.

We would note that using the signal-superposition property of a MAC for direct information fusion from multiple terminals is not new and it has been successfully utilized in solving the classic central estimation officer (CEO) problem in the area of remote estimation of traditional wireless sensor networks (WSNs)~\cite{CEO}. In this application, each sensor takes a noisy observation of the same source (a deterministic parameter or a random process), and concurrently sends the uncoded (linearly scaled) signal to the fusion center through a MAC. The fusion center receives the superimposed signals and reconstructs the source signal of interest.
The CEO problems of deterministic parameter estimation under single and multiple antenna settings were investigated in \cite{Alex} and \cite{Cui}, respectively.
A CEO problem for Gauss-Markov process estimation was considered in \cite{Jiang1}.
Although both the traditional CEO and the recent AirComp systems make use of MACs for data fusion, the former only needs to estimate a single-source signal while the latter has to estimate a function of distributed multi-source signals.
Thus, the design targets and the optimal solutions to these two problems are fundamentally different.

\subsection{Contributions and Paper Organization}
In this paper, we focus on a baseline single-antenna AirComp system with non-zero receiver noise and unequal channel coefficients, where the sensors send linearly scaled pre-processed signals simultaneously to the receiver, which then linearly scales the received signal as the computing output of the sum of the pre-processed signals.\footnote{Such the AirComp system actually uses an analog (or coding-free) transmission method. Note that in the area of WSNs and wireless remote estimation/control systems, extensive research has focused on the fusion of analog data instead of encoded digital data~(see \cite{Jiang1,Alex,Cui}, \cite{AlexControl,liu2019wireless} and references therein). This is mainly because digital transmission achieves an exponentially worse performance than analog signaling in terms of distortion between the source signal and the recovered signal, which has been proved by the pioneer work in~\cite{Gastpar1}.}
The main contributions of the paper are summarized as follows.

{
	$\bullet$ We consider an optimization problem to minimize the computation mean-squared error (MSE) of the sum of the pre-processed signals at the receiver by optimizing the transmitting-receiving (Tx-Rx) scaling policy of a single-antenna AirComp system, where each sensor has a peak power limit for transmission.
	We note that such a type of problem was originally proposed under a more complicated setting of multi-antenna multi-modal sensing in \cite{GuangxuMIMO}, which, however, only gave a suboptimal solution due to the non-convexity of the problem.
	Under the single-antenna setting, although the problem is still non-convex, we derive the optimal solution in closed form.\footnote{We note that the same solution has independently arrived in a parallel work~\cite{CAO}, which will be discussed in the latter part of the introduction.}
	Actually, when considering the single-antenna and single-modal setting, the solution in~\cite{GuangxuMIMO} degrades to a Tx-Rx scaling policy of a channel-inversion type, which leads to a much larger computation MSE than that of the optimal policy obtained in our paper.
	Furthermore, we extend our optimal solution of the single-antenna case to a more practical case of AirComp, where each sensor with simple hardware components has one antenna and the receiver has $N$ antennas.
	Our results show that the solution provides a significant  reduction of computation MSE compared to the existing work~\cite{GuangxuMIMO}.

	$\bullet$ We also investigate a MAC system for distributed remote estimation, which is closely related to the AirComp system.
	In this system, each sensor sends its scaled measurement signal to the receiver simultaneously and the receiver needs to recover each sensor's signal as accurately as possible. We formulate and solve an optimization problem to minimize the sum of the estimation MSE of each sensor's signal at the receiver by optimizing the Tx-Rx scaling policy, under the peak power constraint of each sensor.\footnote{Note that the research of traditional MAC is mostly focused on rate-centric systems, where the design target is commonly to maximize the achievable sum rate of the $K$ users~\cite{MACRate1}. {For estimation-centric MAC, the optimization problem for minimizing the sum of estimation MSEs has not been considered in the open literature. Only its dual system, i.e., an estimation-centric broadcast channel (BC) system, was investigated in \cite{MACMSE}, under a sum power constraint.}}
	Interestingly, such an estimation problem of the MAC system is a sum-of-MSE problem, while the AirComp system introduces an MSE-of-sum problem.
	We also prove the condition under which the optimal Tx-Rx scaling policies of the MAC and AirComp systems are equivalent to each other.

	$\bullet$
	 We investigate the ergodic performance of AirComp systems in terms of the average computation MSE and the average power consumption under Rayleigh fading channels.
	 We comprehensively study the scaling laws of the average computation MSE and the average power consumption of different Tx-Rx scaling policies with respect to the number of sensors $K$. 
	 Also, we define two types of policies: the computation-effective policy, which has a vanishing average computation MSE, and the energy-efficient policy, which has a vanishing average power consumption, with the increasing number of sensors $K$.
	Since there is a tradeoff in policy design between the computation effectiveness and the energy efficiency of the AirComp system,	it is not clear whether there exists a policy that is both computation-effective and energy-efficient.
	We rigorously prove the existence of such the type of policy.	
	Moreover, for the computation-optimal policy, we prove that the policy is a computation-effective one and its average computation MSE has a decay rate of $O(1/\sqrt{K})$. Our numerical results show that the policy is also energy-efficient.

Note that this paper is conducted in parallel and independent of~\cite{CAO}, which has also obtained the same computation-optimal policy in the single-antenna case but with a different proof structure. 
Recall that 
the scaling-law analysis, the tradeoff between the computation effectiveness and the energy efficiency, and the comparison study between the sum-of-MSE (MAC) and the MSE-of-sum (AirComp) problems in our paper have not been considered in the open literature of AirComp including~\cite{GuangxuMIMO} and~\cite{CAO}.
}

{The remainder of the paper is organized as follows: Sec. II describes the AirComp system and formulates a computation-MSE-minimization problem.
In Sec. III, we solve the optimization problem and obtain the computation-optimal Tx-Rx scaling policy.
In Sec. IV, we study a MAC-based remote estimation problem which is closely related to the AirComp problem, and then compare the optimal policies of the two problems.
In Sec. V, the ergodic performance of the AirComp system with different Tx-Rx scaling policies are investigated in terms of the average computation MSE and the average power consumption.
Sec. VI numerically evaluates the performance of AirComp systems with different policies.
Finally, Sec. VII concludes this work.}

Notations:
For two
functions $g(x)$ and $f(x)$, the notations $g(x) = O(f(x))$ and $g(x) = o(f(x)), x\rightarrow \infty$, mean that
$\limsup_{x\rightarrow\infty} g(x)/f(x)<\infty$ and $\lim\limits_{x\rightarrow\infty} g(x)/f(x)=0$, respectively. We denote $f(x) \sim g(x)$ if $g(x) = O(f(x))$ and $f(x) = O(g(x))$.
$\mathbb{R}$ and $\mathbb{R}_0$ denote the set of real number and the set of non-negative real number, respectively. $\mathbb{N}$ and $\mathbb{C}$ denotes the set of positive integers and complex numbers, respectively.
\begin{figure}[t]
	\renewcommand{\captionfont}{\small} \renewcommand{\captionlabelfont}{\small}	
	\centering
	\includegraphics[scale=0.5]{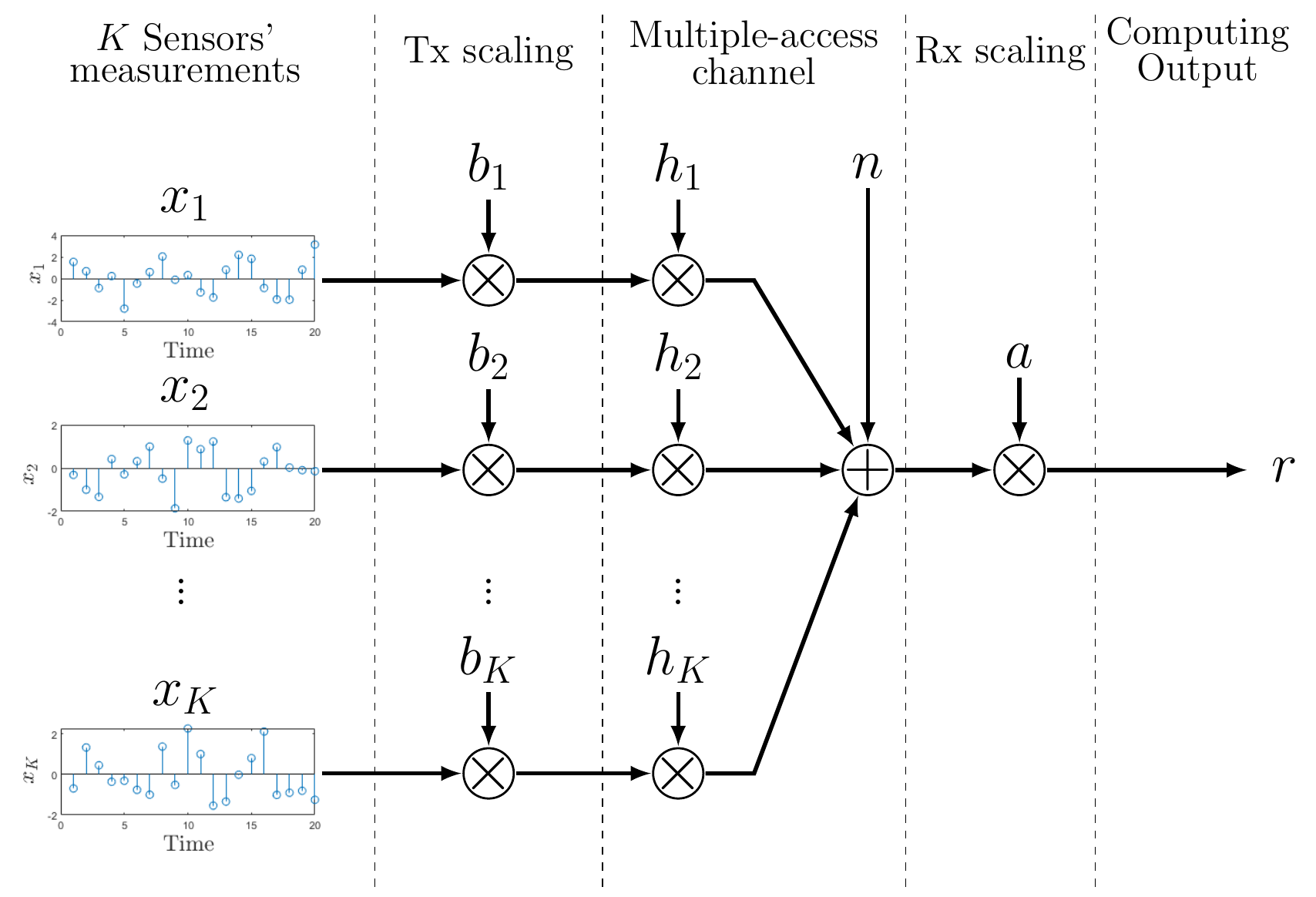}
	\vspace{-0.5cm}
	\caption{Illustration of the AirComp system.}
	\label{fig:system}
	\vspace{-0.7cm}
\end{figure}
\section{System Model}\label{sec:sys}
We consider a $K$-sensor single-antenna AirComp system as illustrated in Fig.~\ref{fig:system}. 
Each sensor's pre-processed signal $x_k\in\mathbb{R},\forall k\in\{1,\cdots,K\}$, is scaled by its Tx-scaling factor $b_k\in\mathbb{C}$ and sent to the receiver simultaneously through a MAC.
The receiver applies a Rx-scaling factor $a$ to the received signal as the computing output of the sum of the $K$ sensors' signals as
\begin{equation} \label{y}
r = a\left(\sum_{k=1}^{K} h_k b_k x_k +n\right),
\end{equation}
{where $h_k\in\mathbb{C}$ is the channel coefficient between sensor $k$ to the receiver and $n$ is the receiver's additive white Gaussian noise (AWGN).
Note that the Rx-scaling factor applied to both the signal and the noise is for the estimation of the computation output $\sum_{k=1}^{K}x_k$, rather than the improvement of signal-to-noise ratio (SNR).}
It is assumed that the channel coefficients are known by both the sensors and the receiver, and the sensors' transmissions are well synchronized~\cite{GoldenbaumWIOPT,GuangxuMIMO,GuangxuWangyong,Gunduz,Osvaldo}.\footnote{The effects of imperfect channel estimation and synchronization have been discussed in \cite{GoldenbaumTCOM,GoldenbaumTSP,Goldenbaumletter}.}
We assume that the pre-processed signal $x_k\in\left[-v,v\right]\subset \mathbb{R}, \forall k\in\{1,\cdots,K\}$, is zero-mean with normalized variance, and is independent with the others~\cite{GuangxuMIMO,Xiaoyang}.

The computation distortion of the ideal sum of the signals $\sum_{k=1}^{K} x_k$ is measured by its MSE~as
\begin{equation}\label{mse}
\mathsf{MSE}=\myexpectation{\lvert r -\sum_{k=1}^{K} x_{k} \rvert^2},
\end{equation}
where the expectation is taken with respect to the randomness of the original signals $\{x_k\}$ and the receiver noise $n$.
Substituting \eqref{y} into \eqref{mse}, the computation MSE of the AirComp system is rewritten as
\begin{equation} \label{MSE}
\mathsf{MSE} = \sum_{k=1}^{K} \vert a h_k b_k -1 \vert^2 + \sigma^2 \vert a \vert^2.
\end{equation}
{Note that the Rx-scaling factor is operated on the digital domain, i.e., the received signal $\left(\sum_{k=1}^{K} h_k b_k x_k +n\right)$ is sampled and quantized before scaling by $a$, and there is no constraint on the Rx-scaling factor.}
Considering a peak power constraint of each sensor's transmission, $P'$, we have
\begin{equation}
\max_{x_k \in \left[-v,v\right]}{\vert b_k x_k\vert^2} =\vert b_k \vert^2 v^2 \leq P' \iff \vert b_k \vert^2 \leq P, \forall k,
\end{equation}
where $P\triangleq P'/v^2$.

Since the pre-processed signal $x_k$ has a normalized variance, the average transmission power of sensor $k$ is $\myexpectation{\vert b_k x_k\vert^2}=\vert b_k \vert^2$. Thus, the power consumption of the $K$-sensor AirComp system~is 
\begin{equation} \label{def_pw}
\pw \triangleq \sum_{k=1}^{K} \vert b_k \vert^2.
\end{equation}

Given the channel coefficients $\{h_k\}$, the MSE-minimization problem in terms of the Tx-Rx scaling policy under the peak power constraint is formulated as
\begin{subequations}\label{problem}
	\begin{alignat}{2}
	&\!\min_{a, \{b_k\}}        &\qquad& \mathsf{MSE}\label{eq:optProb}\\
	&\text{subject to} &      & \vert b_k \vert^2 \leq P, \forall k. \label{eq:constraint1}
	\end{alignat}
\end{subequations}
From the target function definition in \eqref{MSE}, 
given the complex Rx-scaling factor $a$ and the channel coefficient $h_k$, sensor $k$ is always able to adjust the phase of its Tx-scaling factor $b_k$ for phase compensation without changing its magnitude such that the term $ah_kb_k$ is real and non-negative and thus minimizes $\vert ah_kb_k -1 \vert$ in \eqref{MSE}. In this sense, only the magnitudes of $a$, $\{h_k\}$ and $\{b_k\}$ have effect on achieving the minimum $\mse$ in problem~\eqref{problem}. Therefore, without loss of generality, we set $a,h_k\in \mathbb{R}_0,$ and $b_k\in [0,\sqrt{P}],\forall k$, in the rest of the paper.

Problem \eqref{problem} is non-convex as the target function \eqref{eq:optProb} is not a convex function of $a$ and $\{b_k\}$. However, the problem is convex when $a$ or $\{b_k\}$ is fixed. Thus, an intuitive method for solving the problem can be the alternating-direction method, i.e., fixing $\{b_k\}$ that satisfies \eqref{eq:constraint1} and solving the optimal $a$, and then in turn solving the optimal $\{b_k\}$ with the optimal $a$, and so on. However, such the algorithm may not guarantee the convergence to a global optimal solution.

In the following section, we derive the closed-form Tx-Rx scaling policy of problem \eqref{problem}, which is named as the \emph{computation-optimal policy}.

\section{Computation-Optimal AirComp System}\label{sec:opti}
\subsection{Computation-Optimal Policy}\label{subsec:opti}
Without loss of generality, we assume the channel coefficients have the property that $0\triangleq h_0<h_1\leq h_2 \leq \cdots \leq h_K<h_{K+1}\triangleq \infty$.
We then introduce a sequence of $(K+1)$ disjoint intervals $\{\mathcal{S}_k\}$ that covers $\mathbb{R}_0$ as
\begin{equation}\label{sets}
\mathcal{S}_k \triangleq 
\left\lbrace \begin{aligned}
&\left(\frac{1}{h_{1} \sqrt{ P}},\infty\right), k=0,\\
&\left(\frac{1}{h_{k+1} \sqrt{ P}}, \frac{1}{h_{k} \sqrt{P}}\right], k=1,\cdots,K-1,\\
&\left[0, \frac{1}{h_{K} \sqrt{ P}}\right], k=K.
\end{aligned}
\right.
\end{equation}

From \eqref{MSE}, it is clear that if $a\in \mathcal{S}_i$,
where $i$ is named as the \emph{critical number} of the Tx-Rx scaling policy,
$\left\lvert a h_k b_k -1  \right\rvert^2$ monotonically decreases with the increasing of $b_k \in [0,\sqrt{P}], \forall k \in \{1,\cdots,i\}$; while $\left\lvert a h_k b_k -1  \right\rvert^2$ is minimized and equal to zero when $b_k=1/(a h_k) < \sqrt{P}, \forall k \in \{i+1,\cdots,K\}$.
Thus, we have the following result.
\setcounter{section}{1}
\begin{lemma}\label{lem:1}
	\normalfont
	If the Rx-scaling factor $a\in \mathcal{S}_i$, $i=0,1,\cdots,K$, the optimal Tx-scaling factors $\{b_k\}$ are given as
	\begin{equation} \label{optimal_b}
	b_k=
	\left\lbrace
	\begin{aligned}
	&\sqrt{P},&& 1\leq k\leq i \\
	&\frac{1}{a h_{k}},&&  i <k \leq K.
	\end{aligned}
	\right.
	\end{equation}

\end{lemma}
\begin{remark}
Lemma~\ref{lem:1} shows that the critical number $i$ indicates the number of sensors using the maximum power for transmission.
Also, the Tx-scaling factors of the computation-optimal policy has a switching structure, i.e., $i$ sensors with the smallest channel coefficients have to use the maximum power for transmission, while the power consumption of any of the rest $(K-i)$ sensors is of a channel-inversion type.
Also, more sensors have to use the maximum power for transmission if the Rx-scaling factor is small.
\end{remark}

Using Lemma~\ref{lem:1}, 
since $ \sum_{k=i+1}^{K} \vert a h_k b_k -1 \vert^2 =0$, the target function of problem \eqref{problem} can be rewritten as
\begin{equation}\label{target_temp}
	\mse = \sum_{k=1}^{i} \left\lvert a h_k \sqrt{P} -1  \right\rvert^2  + \sigma^2 \vert a \vert^2, a\in \mathcal{S}_i,	
\end{equation}
which is a quadratic function of the Rx-scaling factor $a$. Then, we can directly obtain the following result.
\begin{lemma} \label{lem:2}
	\normalfont	
{Given the constraint that the Rx-scaling factor belongs to $\mathcal{S}_i$, $i=0,1,\cdots,K$, the optimal Rx-scaling factor $a_i \in \mathcal{S}_i$ is given as}
	\begin{equation}\label{opti_a}
	a_i \triangleq \left\lbrace 
	\begin{aligned}
	&\left(\frac{1}{h_{i+1} \sqrt{P}}\right)^+,&&\text{if } g_i < \mathcal{S}_i\\
	&g_i,&&\text{if } g_i \in \mathcal{S}_i\\
	&\frac{1}{h_{i} \sqrt{P}},&&\text{if } g_i > \mathcal{S}_i
	\end{aligned}
	\right.
	\end{equation}
	where the operator $(u)^+$ indicates approaching to the real number $u$ from right side, 
	and with a bit abuse of notation, $g_i > \mathcal{S}_i$ and $g_i < \mathcal{S}_i$ denote $g_i > \frac{1}{h_{i} \sqrt{P}}$ and $g_i \leq \frac{1}{h_{i+1} \sqrt{P}}$, respectively,		
	where $g_i$ is the optimal solution of \eqref{target_temp} without the constraint $a\in\mathcal{S}_i$ and is given as
	\begin{equation} \label{g_i}
	g_i \triangleq \begin{cases}
	0,&i=0\\
	\frac{\sqrt{P} \sum_{k=1}^{i}h_k}{\sigma^2 + P \sum_{k=1}^{i}h^2_k},&i=1,\cdots,K.
	\end{cases}
	\end{equation}

\end{lemma}

Based on Lemma~\ref{lem:2}, if $g_i\notin \mathcal{S}_i$, i.e., the first and the last cases in \eqref{opti_a}, the following property reveals how to find a better Rx-scaling factor $a$ leading to a smaller $\mse$.
\begin{lemma}\label{lem:3}
	\normalfont
	If $g_i < \mathcal{S}_i$ and $i<K$, there exists $a\in \mathcal{S}_{i+1}$ that achieves a smaller $\mse$ than $a_i$.
	If $g_i>\mathcal{S}_i$ and $i>0$, there exists $a\in \mathcal{S}_{i-1}$ that achieves a smaller $\mse$ than $a_i$.
\end{lemma}
\begin{proof}
	For the first case, it is not hard to see $a=\frac{1}{h_{i+1}\sqrt{P}}\in \mathcal{S}_{i+1}$ leads to a smaller $\mse$ due to the continuity of the target function \eqref{target_temp}. For the second case,  assuming that $g_i\in \mathcal{S}_{i'}$ and $i'<i$, since the quadratic function \eqref{target_temp} monotonically decreases and then increases, there exists $a'\in \mathcal{S}_{i-1}$ such that 
	\begin{equation} \label{ineq}
\begin{aligned}
\sum_{k=1}^{i-1} \left\lvert a' h_k \sqrt{P} -1  \right\rvert^2  + \sigma^2 \vert a' \vert^2 &\leq \sum_{k=1}^{i} \left\lvert a' h_k \sqrt{P} -1  \right\rvert^2  + \sigma^2 \vert a' \vert^2 
&< \sum_{k=1}^{i} \left\lvert a h_k \sqrt{P} -1  \right\rvert^2  + \sigma^2 \vert a \vert^2,
\end{aligned}
	\end{equation}
	where the first term in \eqref{ineq} is the minimum $\mse$ achieved by $a'\in \mathcal{S}_{{i-1}}$ from Lemma~\ref{lem:1}, completing the proof.
\end{proof}
Jointly using Lemma~\ref{lem:2} and Lemma~\ref{lem:3}, it can be obtained that
\begin{corollary} \label{cor:1}
	\normalfont
	The optimal Rx-scaling factor $a \notin \mathcal{S}_0$, i.e., in order to achieve the minimum $\mse$, at least one sensor needs to transmit with the maximum power.
\end{corollary}

Based on Lemma~\ref{lem:1} and Lemma~\ref{lem:2}, to find the optimal Rx-scaling factor $a^\star$, we only need to know the optimal critical number $i^\star$, i.e., $a^\star =a_{i^\star} \in \mathcal{S}_{i^\star}$, which is the most important part for solving the problem. Thus, problem \eqref{problem} is reformulated as 
	\begin{align} \label{problem2}
	&\!\min_{1 \leq i \leq K} \qquad \mse_i \triangleq \sum_{k=1}^{i} \left\lvert a_i h_k \sqrt{P} -1  \right\rvert^2  + \sigma^2 \vert a_i \vert^2\\ \label{test2}
	&\text{subject to} \quad \eqref{opti_a}.
	\end{align}


From Lemma~\ref{lem:2}, since $a_i$ depends on $g_i,\forall i \in \{1,\cdots,K\}$, the optimal $a$, $a^\star 
\in \{a_1,\cdots,a_K\}$, depends on the sequence $\{g_i\}$. In what follows, we introduce the technical lemmas of the properties of $\{g_i\}$, which will be utilized for finding $a^\star$.

\setcounter{section}{2}
\setcounter{lemma}{0}
\begin{lemma}[Switching structure] \label{lem:2.1}
	\normalfont
If $g_i \in \mathcal{S}_i$, then
\begin{equation}\label{g_results}
\left\lbrace \begin{aligned}
&g_{i+1} > \mathcal{S}_{i+1}, &&\text{ if } i<K,\\
&g_{i-1} < \mathcal{S}_{i-1}, &&\text{ if } i>0.\\
\end{aligned}
\right.
\end{equation}
\end{lemma}
\begin{proof}
\eqref{g_results} can be verified by using the definition of $g_i$ in \eqref{g_i} and the property that $\frac{a}{b}\geq \frac{c}{d} \iff \frac{a+c}{b+d} \geq \frac{c}{d}$, where $a,b,c,d>0$. The following lemmas can be verified using the similar steps, and the proofs are omitted for brevity.
\end{proof}

\begin{lemma}[Consistency]\label{lem:2.2}
	\normalfont	
	If $g_i < \mathcal{S}_i$ and $i>0$, then
$g_{i-1} < \mathcal{S}_{i-1}$.	
	If $g_i > \mathcal{S}_i$ and $i<K$, then
	$g_{i+1} > \mathcal{S}_{i+1}$.
\end{lemma}

\begin{lemma}[Monotonicity]\label{lem:2.3}
	\normalfont	
	If $g_i < \mathcal{S}_i$ and $i<K$, then
	$g_{i} \leq  g_{i+1}$.
	If $g_i > \mathcal{S}_i$ and $i>0$, then
	$g_{i} \leq g_{i-1}$.
\end{lemma}
\setcounter{section}{3}

Jointly using Lemma~\ref{lem:2.1}, Lemma~\ref{lem:2.2} and Lemma~\ref{lem:2.3}, it  shows the unimodality of the sequence $\{g_i\}$, i.e., there exists $i^\star$ such that $g_i$ monotonically increases and decreases with $i$ for all $i\leq i^\star$ and $i\geq i^\star$, respectively.
Then, using Lemma~\ref{lem:2.1} and Lemma~\ref{lem:2.2} together with Lemma~\ref{lem:3}, the unimodality of the sequence $\{-\mse_i\}$ in \eqref{problem2} can be easily verified,
as $\mse_i$ monotonically decreases and increases with $i$ for $i\leq i^\star$ and $i\geq i^\star$, respectively.
Therefore, the optimal Tx-Rx scaling policy is given below.
\begin{theorem}[Computation-optimal policy]\label{theory:optimal}
	\normalfont
The optimal critical number of problem \eqref{problem} is
\begin{equation} \label{opt_i}
i^\star = \arg\max_{1 \leq i \leq K} g_i.
\end{equation}
The optimal Rx-scaling factor $a^\star$ is $a_{i^\star}$ by taking $i^\star$ into \eqref{opti_a}.
The optimal Tx-scaling factors $\{b_k^\star\}$ are obtained by taking $i^\star$ into \eqref{optimal_b}.
\end{theorem}

{We note that this result is independently arrived at~\cite{CAO} with a different proof.}

As a consequence of Theorem~\ref{theory:optimal}, the computation MSE and the power consumption induced by the computation-optimal policy are obtained below.
\begin{proposition}	
	\normalfont
	The minimum computation MSE of the AirComp system under the peak power constraint is
	\begin{equation}\label{mse_optimal}
	\mathsf{MSE}^\star = \sum_{k=1}^{i^\star} \left( a_{i^\star} h_k \sqrt{P}-1\right)^2 + \sigma^2 (a_{i^\star})^2.
	\end{equation}
	The power consumption of the AirComp system  induced by the computation-optimal policy is
	\begin{equation}\label{power}
	\mathsf{PW} 
	= \sum_{k=1}^{K} {b^{\star}_k}^2
	=P i^\star + \frac{1}{(a_{i^\star})^2} \sum_{k=i^\star +1}^{K} \frac{1}{h^2_k}.
	\end{equation}
\end{proposition}

{
\subsection{Extension to Multi-Antenna Receiver Case} \label{sec:multi}
We extend our solution of the single-antenna AirComp system to a more practical multi-antenna receiver case, where each sensor with simple hardware components has one antenna and the receiver has $N$ antennas.
The computation MSE is rewritten as
\begin{equation}\label{mse_multi}
\mse=  \sum_{k=1}^{K} \vert \mathbf{a}^\top \mathbf{h}_k b_k -1 \vert^2 + \sigma^2 \vert \mathbf{a} \vert^2,
\end{equation}
where $\mathbf{a} \in \mathbb{C}^{N \times 1}$ is the  Rx-scaling vector at the receiver, and $\mathbf{h}_k \in \mathbb{C}^{N \times 1}$ is the channel-coefficient vector between sensor $k$ and the receiver. $(\cdot)^\top$ is the matrix-transpose operator.
Taking \eqref{mse_multi} into \eqref{problem}, such an computation-MSE-minimization problem in the multi-antenna receiver case is intractable and NP-hard~\cite{GuangxuMIMO}. Thus, we develop algorithms to find the near optimal solutions based on the single-antenna solution in Theorem~\ref{theory:optimal}.

By denoting $\mathbf{a} \triangleq \vert \mathbf{a} \vert \mathbf{v}$, where $\mathbf{v}$ is a unit vector, and from the discussion below \eqref{problem}, \eqref{mse_multi} is equivalent to
\begin{equation}\label{mse_multi2}
\mse=  \sum_{k=1}^{K} \bigg\vert \vert \mathbf{a} \vert \big\vert \mathbf{v}^\top \mathbf{h}_k \big\vert  b_k -1 \bigg\vert^2 + \sigma^2 \vert \mathbf{a} \vert^2.
\end{equation}
It is clear that if $\mathbf{v}$ is given, the optimal $\vert \mathbf{a} \vert$ and $\{b_k\}$ can be solved directly by replacing $h_k$ with $\vert \mathbf{v}^\top \mathbf{h}_k \vert$ in Theorem~\ref{theory:optimal}. Although finding the optimal $\mathbf{v}$ is not tractable, two easy-to-computation algorithms can find the suboptimal unit vectors: one is the optimal antenna selection, which finds the optimal $\mathbf{v}$ that belongs to the set of $\{[1,0,\cdots,0]^\top,[0,1,0,\cdots,0]^\top, \cdots,[0,\cdots,0,1]^\top\}$; the other is to randomly generate a sequence of $\mathbf{v}$ and find the optimal one that induces the minimum $\mse$ in \eqref{mse_multi2}.
In Sec. VI, we will show that these two methods can achieve  much better computation MSEs than the method in~\cite{GuangxuMIMO}.
}

\subsection{Benchmark Policies}\label{sec:benchmark}
The computation-optimal policy in Theorem~\ref{theory:optimal} 
needs to first sort the $K$ channel coefficients (e.g., using an insertion-sort algorithm), which has a computation complexity of $O(K^2)$, and then compute the largest parameter $g_i$, which is a non-linear function of $i$ channel coefficients. Thus, we also present two intuitive and easy-to-compute benchmark policies of AirComp systems for comparison.

\begin{definition}[Channel-inversion policy~\cite{GuangxuMIMO,Xiaoyang}]
	\normalfont
The channel-inversion policy has the critical number $i=1$, i.e., the Rx-scaling factor $a\in \mathcal{S}_1$, and the Tx and Rx-scaling factors are given as
\begin{equation}\label{benchmark1}
	\begin{aligned}
&b_k= \sqrt{P} \frac{h_1}{h_k}, \forall k,
&a=\frac{1}{\sqrt{P}h_1}.
\end{aligned}
\end{equation}
\end{definition}

The channel-inversion policy is commonly considered in the literature of AirComp~\cite{GuangxuMIMO,Xiaoyang}, which guarantees that the computing output $r$ in \eqref{y} is an unbiased estimation of the sum of the original signals, i.e.,
\begin{equation}
\myexpectation{r-\sum_{k=1}^{K} x_k \bigg\vert x_1,\cdots,x_K} = a\myexpectation{n}=0.
\end{equation}

\begin{definition} [Energy-greedy policy] \label{def:energy_greedy}
	\normalfont
	The energy-greedy policy always chooses the critical number $i=K$, i.e., the Rx-scaling factor $a\in \mathcal{S}_K$, and the Tx- and Rx-scaling factors according to Lemma~\ref{lem:1} and Lemma~\ref{lem:2} are given as
\begin{subequations}\label{benchmark2}
		\begin{alignat}{2} 
		\label{benchmark2_2}
		&b_k= \sqrt{P}, \forall k,\\		
		\label{benchmark2_1}
	&a= \min \left\lbrace \frac{1}{\sqrt{P} h_K}, \frac{\sqrt{P}\sum_{k=1}^{K}h_k}{\sigma^2+P \sum_{k=1}^{K}h^2_k} \right\rbrace.
	\end{alignat}
\end{subequations}
\end{definition}

The energy-greedy policy requires all the sensors to always transmit with the maximum power regardless of the channel conditions. 
Since the optimal critical number of an AirComp system, $i^\star$, usually takes a value between $1$ and $K$, the channel-inversion policy and the energy-greedy policy can be treated as two extreme cases.

%
%

\section{AirComp System versus Traditional MAC System} \label{sec:comparison}
{In this section, we consider a remote estimation problem based on a MAC system, which is closely related to the AirComp problem since both the systems leverage the superposition property of the MAC channel.}
The MAC system consists of $K$ sensors and one receiver, where the receiver has to recover every sensor's signal $x_k,\forall k\in\{1,\cdots,K\}$. 
The sensors adopt the coding-free transmission method same as the AirComp system, i.e., each of the $K$ sensors scales its measurement signal by the Tx-scaling factor $\tilde{b}_k$ and simultaneously sends it to the receiver. Then, the receiver estimates each of the original signals by linearly scaling the received signal with the Rx-scaling factor $\tilde{a}_k,\forall k\in\{1,\cdots,K\}$.
The estimated sensor $k$'s signal is given as
\begin{equation}
r_k=\tilde{a}_k \left(\sum_{k=1}^{K} h_k \tilde{b}_k x_k +n\right),\forall k\in\{1,\cdots,K\}.
\end{equation}
With a bit abuse of notation, the estimation MSE of sensor $k$ is denoted and obtained as
\begin{align}\label{msek_2}
&\mse_k  \triangleq \myexpectation{\lvert r_k - x_k \rvert^2}
= (\tilde{a}_k h_k \tilde{b}_k-1)^2 + \tilde{a}^2_k \!\! \sum_{j\in \{1,\cdots,K\}\backslash k} (h_j \tilde{b}_j)^2 +\tilde{a}^2_k \sigma^2.
\end{align}

\subsection{Optimal Sum of MSE}
We aim to design the optimal Tx-Rx scaling policy to minimize the sum of the estimation MSEs of the $K$ sensors' signals under the individual power constraint, and have the following problem:
\begin{subequations}\label{problem3}
	\begin{alignat}{2}
	&\!\min_{\{\tilde{a}_k\}, \{\tilde{b}_k\}}        &\qquad& \sum_{k=1}^{K}\mse_k\label{eq:optProb3}\\
	&\text{subject to} &      & \vert \tilde{b}_k \vert^2 \leq P, \forall k. \label{eq:constraint3}
	\end{alignat}
\end{subequations}
\begin{remark}
This optimal estimation problem of the MAC system is a sum-of-MSE problem, while the optimization problem of the AirComp system in Sec.~\ref{subsec:opti} can be treated as an MSE-of-sum (signal) problem. Thus, the optimal Tx-Rx scaling policies of the two problems are different in general.
\end{remark}

{When $\{\tilde{b}_k\}$ are given and satisfy the constraint \eqref{eq:constraint3}, the target function \eqref{eq:optProb3} is a quadratic function of $\{\tilde{a}_k\}$, and it is clear that the variable $\tilde{a}_k$ only has an effect on $\mse_k$ rather than the other MSE objects.}
The optimal Rx-scaling factor for sensor $k$ is obtained directly as
\begin{equation} \label{ak}
\tilde{a}_k=\frac{h_k \tilde{b}_k}{\sigma^2 + \sum_{j=1}^{K}(h_j \tilde{b}_j)^2}.
\end{equation}
{Taking \eqref{ak} into \eqref{msek_2}, which does not change the optimality compared to the original problem~\eqref{problem3}, we have}
\begin{equation}\label{mse_k}
\mse_k = 1- \frac{(h_k \tilde{b}_k)^2}{\sigma^2 + \sum_{j=1}^{K}(h_j \tilde{b}_j)^2},
\end{equation}
and thus
\begin{equation}
\sum_{k=1}^{K} \mse_k = K -\frac{ \sum_{j=1}^{K}(h_j \tilde{b}_j)^2}{\sigma^2 + \sum_{j=1}^{K}(h_j \tilde{b}_j)^2},
\end{equation}
which is a monotonic decreasing function of $\{ \tilde{b}^2_k \}$.

 Therefore, the optimal solution of problem~\eqref{problem3} is given below.
\begin{theorem}[Optimal sum-of-MSE  policy]\label{theory:mac}
\normalfont
The optimal Tx-Rx scaling policy of problem~\eqref{problem3} of the MAC system is given as
\begin{align}
\tilde{a}^\star_k&=\frac{h_k \sqrt{P}}{\sigma^2 + P \sum_{j=1}^{K}(h_j)^2}, 
\ \tilde{b}^\star_k =\sqrt{P}, \forall k.
\end{align}
\end{theorem} 
\begin{remark}[When the sum-of-MSE problem is equivalent to the MSE-of-sum problem?]
Different from the computation-optimal policy of the AirComp system in Theorem~\ref{theory:optimal}, where the optimal Tx-scaling factor of each sensor depends on $K$ channel coefficients, the optimal policy of the MAC system has the identical Tx-scaling factor $\sqrt{P}$, i.e., a type of energy-greedy policy that requires each sensor to use the maximum power for transmission.

However, from Theorems~\ref{theory:optimal} and~\ref{theory:mac}, if the sequence $\{g_i\}$ defined in \eqref{g_i} satisfies the condition 
$
g_K=\max_{1 \leq k \leq K} g_k,
$
the optimal Tx-scaling factors of the AirComp system are identical to that of the MAC system, i.e., $b^\star_k=\tilde{b}^\star_k=\sqrt{P},\forall k$. Also the optimal Rx-scaling factor of the AirComp system $a^\star= \frac{\sum_{k=1}^{K} h_k \sqrt{P}}{\sigma^2 + P \sum_{j=1}^{K}(h_j)^2}$ is equal to the sum of the optimal Rx-scaling factors of the MAC system  $\sum_{k=1}^{K} \tilde{a}^\star_k$, and thus the estimation of the sum of the signals is equal to the sum of the estimation of each individual signal, i.e., $r=\sum_{k=1}^{K} r_k$. 
In this sense, the optimal sum-of-MSE policy is equivalent to the optimal MSE-of-sum policy.

{Such a comparison study between the recently proposed AirComp system and the conventional MAC system is very important for having a better understanding of AirComp. In particular, when the condition mentioned earlier is satisfied, the AirComp system simply degrades to the conventional MAC system (i.e., one can estimate the individual signals first and then calculate the sum of all the estimates); otherwise, the AirComp system is distinguished from the MAC system.}
\end{remark}

\subsection{Achievable MSE Region and the Pareto Front}
In addition to the sum of the estimation MSEs of $K$ sensors' signals, we also care about the capability of the MAC system in  providing the estimation quality of the $K$ sensors, which is 
captured by the achievable MSE region defined below. 
\begin{definition}[Achievable MSE region]
\normalfont
Given the channel coefficients $\{h_k\}$ and the individual power constraint $P$, the achievable MSE region of a $K$-sensor MAC, $\mathcal{M}$, is defined by the set of all tuples $(\mse_1,\cdots,\mse_K)$, where $\mse_k$ is defined in \eqref{msek_2}, $\tilde{a}_k$ is given in \eqref{ak}, and $\tilde{b}_k\in [0, \sqrt{P}],\forall k\in\{1,\cdots,K\}$.
\end{definition}

Using \eqref{mse_k}, $b^2_k$ can be represented by $\{\mse_j\}$ as
\begin{equation} \label{s-x}
b^2_k= \frac{\sigma^2}{h^2_k}\frac{1-\mse_k}{\sum_{j=1}^{K} \mse_j -(K-1)},\forall k.
\end{equation}
Since $b^2_k \leq P$ and $\mse_k \leq 1$ in \eqref{mse_k}, the achievable MSE region can be derived as below.
\begin{proposition}
\normalfont
The achievable MSE region of a $K$-sensor MAC, $\mathcal{M}$, satisfies
\begin{subequations} \label{region}
\begin{alignat}{2} \label{region1}
&\mse_k + \frac{P h^2_k}{\sigma^2} \left(\sum_{j=1}^{K} \mse_j - (K-1)\right) \geq 1, \forall k\\ \label{region2}
&\mse_k \leq 1, \forall k.
\end{alignat}
\end{subequations}
\end{proposition}
\begin{remark}
It can be verified that the achievable MSE region \eqref{region} is convex, and \eqref{region1} and \eqref{region2} define the inner and outer boundaries of the region, respectively, as illustrated in Fig.~\ref{fig:region}. Specifically, the inner boundary \eqref{region1} is achieved by letting the $k$th sensor transmit with the maximum power, i.e., $b^2_k=P$.
\end{remark}
\begin{figure}[t]
	\renewcommand{\captionfont}{\small} \renewcommand{\captionlabelfont}{\small}	
	\centering
	\includegraphics[scale=0.9]{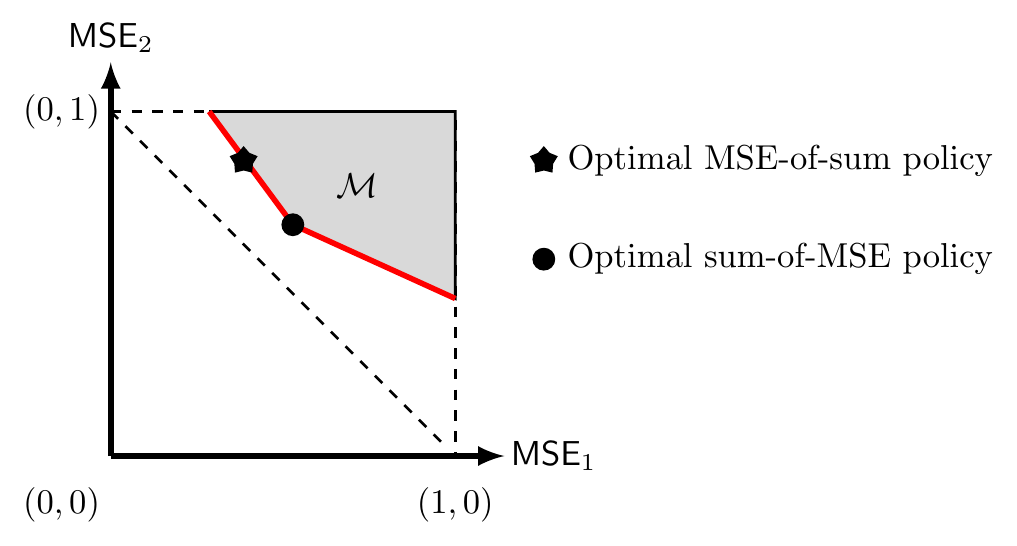}
	\vspace{-0.5cm}	
	\caption{Achievable MSE region of a two-sensor MAC system, where the red broken line is the Pareto front of the sum-of-MSE minimization problem, and the dashed diagonal line indicates that the sum of MSE is larger than $(K-1)=1$, which can be obtained from \eqref{region1} and \eqref{region2}, and the horizontal and vertical dashed lines indicate the outer boundaries of the achievable MSE region, i.e., $\mse_k \leq 1$.}
	\label{fig:region}
	\vspace{-0.5cm}	
\end{figure}
\begin{remark}
	
It can be observed that the Pareto front of the problem~\eqref{problem3} locates on the boundary of $\mathcal{M}$, where it is not possible to reduce an MSE object without increasing the others. 
In particular, we see that the inner boundaries \eqref{region1} bound $\{\mse_k\}$ away from zero. When the equality in \eqref{region1} holds, to reduce $\mse_k,\forall k =1,\cdots,K$, one needs to increase at lease one of the others.
Therefore, by defining the sets
\begin{equation}
\mathcal{B}_k = \left\lbrace (\mse_1,\cdots,\mse_K):  \mse_k + \frac{P h^2_k}{\sigma^2} \left(\sum_{j=1}^{K} \mse_j - (K-1)\right) = 1\right\rbrace , \forall k, 
\end{equation}
the Pareto front of problem~\eqref{problem3} is 
$
\mathcal{P} \triangleq \left(\mathcal{B}_1 \cup  \mathcal{B}_2 \cup \cdots \cup \mathcal{B}_K\right) \cap \mathcal{M},
$
which is illustrated by the red broken line in Fig. 2.
Furthermore, it can be verified that the MSE tuple induced by the optimal sum-of-MSE solution, i.e., $\mse_k = 1- \frac{(h_k)^2 P}{\sigma^2 + \sum_{j=1}^{K}(h_j)^2 P}$, locates at the intersection of the inner boundaries~\eqref{region1}, i.e., $\left(\mathcal{B}_1 \cap  \mathcal{B}_2 \cap \cdots \cap \mathcal{B}_K\right)$, which belongs to the Pareto front $\mathcal{P}$ and hence verifies the correctness of the solution.
\end{remark}

From Corollary~\ref{cor:1}, the computation-optimal policy of the AirComp system assigns at least one sensor using the maximum power for information transmission.
Thus, applying the optimal Tx-scaling factors of the AirComp system $\{b^\star_k\}$ to the MAC system, i.e., letting $\tilde{b}_k=b^\star_k,\forall k$, the achievable MSE tuple $(\mse_1,\cdots,\mse_K)$ falls on the inner boundaries of the achievable MSE region $\mathcal{M}$  as illustrated in Fig.~\ref{fig:region}.
Recall that for the optimal Tx-scaling factors of the MAC system, i.e., $\tilde{b}^\star_k = \sqrt{P}$, the equalities of the $K$ constraints~\eqref{region1} hold, thus the optimal achievable MSE tuple $(\mse_1,\cdots,\mse_K)$ is the intersection of the inner boundaries as illustrated in Fig.~\ref{fig:region}.

\section{Ergodic Performance of AirComp: Computation Effectiveness versus Energy Efficiency}\label{sec:scaling1}
In Sec.~\ref{sec:opti}, the performance of the AirComp system under instantaneous channel conditions has been investigated.
In this section, we focus on the ergodic performance of AirComp systems with different Tx-Rx scaling policies under Rayleigh fading channels,
where each channel coefficient $h_k,\forall k\in\{1,\cdots,K\}$, independently varies with time and has the standard Rayleigh stationary distribution~\cite{GuangxuInference}.
In particular, we investigate the average computation MSE and the average power consumption of the AirComp system defined below.\footnote{Although the following analysis are for AirComp systems with Rayleigh-distributed channel coefficients, the analysis framework can be applied to the cases with other channel-coefficient distributions.}

%

\begin{definition} \label{def:sumpower}
	\normalfont
	The average computation MSE and the average power consumption of an AirComp system are defined respectively as 
	${\myexpectation{\mse}}/{K}$ and ${\myexpectation{\pw}}/{K}$, where $\mse$ and $\pw$ are given in \eqref{MSE} and \eqref{def_pw}, respectively, and the expectation is taken with respect to $K$ random channel coefficients.
\end{definition}
Thus, the average computation MSE and the average power consumption indicate how the computation accuracy/performance and the total power consumption of the AirComp system scale with the increasing computation load, i.e., the increasing $K$.

\begin{definition} \label{def:performance-effective}
	\normalfont
	A Tx-Rx scaling policy of the AirComp system is a computation-effective policy~iff
	\begin{equation}
	\lim\limits_{K\rightarrow \infty} \frac{\myexpectation{\mse}}{K} = 0.
	\end{equation} 
	The policy is an energy-efficient policy iff
	\begin{equation}
	\lim\limits_{K\rightarrow \infty} \frac{\myexpectation{\mathsf{PW}}}{K} = 0.
	\end{equation} 
\end{definition}

For a computation-effective policy, the average computation MSE approaches to zero, while for an energy-efficient policy, the average power consumption approaches to zero, with the increasing computation load.
Therefore, it is interesting to see whether the benchmark policies (including the channel-inversion policy and the energy-greedy policy), and the computation-optimal policy in Sec.~\ref{sec:opti} are computation-effective or energy-efficient or both.

In the following analysis, we denote the channel power gains as $U_k\triangleq h^2_k$, where $U_{k_1}\leq U_{k_2},\forall 1\leq k_1\leq k_2\leq K$. In other words, $\{U_k\}$ are the order statistics of $K$ independent random samples from the standard exponential distributions.

\subsection{Benchmark Policy 1: Channel-Inversion Policy}
\subsubsection{Average Computation MSE}
Taking \eqref{benchmark1} into \eqref{MSE}, the average computation MSE can be derived as
\begin{equation}\label{average_mse}
\frac{\myexpectation{\mathsf{MSE}}}{K} = \frac{1}{K}\myexpectation{\frac{\sigma^2}{P U_1}}.
\end{equation}

Using the property of order statistics~\cite{orderstatistics}, the probability density function (pdf) of the minimum value of $K$ sample from the standard exponential distribution is given as
\begin{equation}\label{pdf_1}
f_{U_1}(u_1) = 
\left\lbrace 
\begin{aligned}
&K e^{-K u_1}, &&u_1\geq 0\\
&0,&& u_1<0
\end{aligned}
\right.
\end{equation}
Thus, $U_1$ follows an exponential distribution and $1/U_1$ follows an inverse exponential distribution with 
the pdf
\begin{equation} \label{pdf1}
f_{1/U_1}(y)= \exp\left(- \frac{K}{y}\right)\frac{K}{y^2}.
\end{equation}
Taking \eqref{pdf1} into \eqref{average_mse}, the average computation MSE is calculated as
\begin{equation}
\frac{\myexpectation{\mse}}{K}=\frac{\sigma^2}{P}\frac{1}{K}\int_{0}^{\infty} y \exp\left(- \frac{K}{y}\right)\frac{K}{y^2} \mathrm{d}y = \infty.
\end{equation}
Therefore, the channel-inversion method actually leads to a poor computation performance, and we have the following result.
\begin{corollary}\label{cory:1}
	\normalfont
	The average computation MSE of the channel-inversion policy is infinite. The policy is not a computation-effective one.
\end{corollary}

\subsubsection{Average Power Consumption}
From \eqref{benchmark1}, the average power consumption is derived as
\begin{equation} \label{ave_power_1}
\frac{\myexpectation{\pw}}{K}=\frac{P}{K} \left(1+\myexpectation{\sum_{k=2}^{K} \frac{U_1}{U_k} }\right),
\end{equation}
where
\begin{equation}
\begin{aligned}
&\myexpectation{\sum_{k=2}^{K} \frac{U_1}{U_k} }
=\int_{u_1,\cdots,u_K}^{}\left(\sum_{k=2}^{K} \frac{u_1}{u_k}\right)
f(u_1,u_2,\cdots,u_K) \mathrm{d} u_1,\cdots,\mathrm{d} u_K,
\end{aligned}
\end{equation}
and $f(u_1,u_2,\cdots,u_K)$ is the joint distribution of $U_1,\cdots,U_K$.
Then, we have the following result.
\begin{corollary}\label{cory:2}
	\normalfont
The average power consumption of the channel-inversion policy is
\begin{equation}
\frac{\myexpectation{\pw}}{K} = \frac{P\ln K}{K-1},
\end{equation}
which has the scaling law as
\begin{equation}
\frac{\myexpectation{\pw}}{K} \sim \frac{\ln K}{K},K\rightarrow \infty.
\end{equation}
The policy is an energy-efficient one.
\end{corollary}
\begin{proof}
	See Appendix A
\end{proof}

\subsection{Benchmark Policy 2: Energy-Greedy Policy}
\subsubsection{Average Computation MSE}
From \eqref{benchmark2},
the Rx-scaling factor of the energy-greedy policy may have the sum of channel coefficients and the sum of channel power gains in the numerator and the denominator, respectively, which makes the analysis of average computation MSE difficult. Thus, we analyze an upper bound of the average computation MSE.
From Lemma~\ref{lem:2}, letting $a = \frac{1}{\sqrt{P}h_K}$ always results in an $\mse$ no smaller than that in \eqref{benchmark2_1}, and thus we have
\begin{align}\label{upperbound2}
\frac{\myexpectation{\mse}}{K} 
&\leq \frac{1}{K} \myexpectation{\sum_{k=1}^{K} \left\lvert 1- \frac{h_k}{h_K} \right\rvert^2 +\frac{\sigma^2}{P h^2_K}}
\leq 1 + \myexpectation{\frac{\sigma^2}{P K U_K}},
\end{align}
where the last inequality in \eqref{upperbound2} is due to the fact that $h_k \leq h_K,\forall k$.

Again, using the property of order statistics~\cite{orderstatistics}, the largest sample of $K$ standard exponential distribution $U_K$ has the pdf as
\begin{equation}\label{pdf_K}
f_{U_K}(u_K) = \begin{cases}
K e^{-u_K} (1-e^{-u_K})^{K-1},&u_K \geq 0\\
0, &u_K < 0
\end{cases} 
\end{equation}
Applying \eqref{pdf_K} onto \eqref{upperbound2}, 
it is obtained as 
\begin{equation}
\begin{aligned}
\frac{\myexpectation{\mse}}{K} 
&\leq 1+ \int_{0}^{\infty} \frac{e^{-x} (1-e^{-x})^{K-1}}{x} \mathrm{d}x
\leq 1+ \int_{0}^{\infty} \frac{e^{-x} (1-e^{-x})}{x} \mathrm{d}x 
= 1+\ln 2 < \infty.
\end{aligned}
\end{equation}
Also, it can be directly proved that $\limsup_{K\rightarrow \infty} \myexpectation{\mse}/K$ is greater than a positive constant by using Theorem~\ref{theory:first_i} from the latter part of the paper.
We have the following result.
\begin{corollary}\label{cor:benchmark2-1}
	\normalfont
The average computation MSE of the energy-greedy policy is upper bounded by $1+\ln 2$,
and the scaling law of the average computation MSE can be written as
\begin{equation}
\frac{\myexpectation{\mse}}{K}\sim 1, K\rightarrow \infty.
\end{equation}
This policy is not a computation-effective one.
\end{corollary}


\subsubsection{Average Power Consumption}
From \eqref{benchmark2_2}, each sensor uses the same power $P$ for information transmission, and we have the result below.
\begin{corollary} \label{cor:benchmark2-2}
	\normalfont
The average power consumption of the energy-greedy policy~is	
\begin{equation}
\frac{\myexpectation{\pw}}{K} = \lim\limits_{K\rightarrow \infty} \frac{\myexpectation{\pw}}{K} = P \neq 0.
\end{equation}
This policy is not an energy-efficient one.
\end{corollary}

Comparing Corollaries~\ref{cor:benchmark2-1} and~\ref{cor:benchmark2-2} with Corollaries~\ref{cory:1} and~\ref{cory:2}, the energy-greedy policy provides a better computation performance but has a lower energy efficiency than the channel-inversion policy when $K$ is large.
Therefore, there exists a design tradeoff between the computation effectiveness and the energy efficiency, and it is important to see whether there exists a Tx-Rx scaling policy that is both computation-effective and energy-efficient.

%
%

\subsection{The Existence of Computation-Effective and Energy-Efficient Policies}\label{sec:scaling2}
We introduce a new type of policy and study its scaling laws in terms of average computation MSE and average power consumption, which will further shed lights on the existence of computation-effective and energy-efficient policies and the scaling laws of the computation-optimal policy.
\subsubsection{Construction of A New Policy}

\begin{definition}[First-$\imath$ policy] \label{def:policy}
	\normalfont
	A Tx-Rx scaling policy of the AirComp system is a first-$\imath$ policy if it satisfies:\\
	i) the critical number is determined by a function, i.e., $i=\imath(K)$, where $\imath: \mathbb{N} \rightarrow \mathbb{N}$ and $\imath(K)\leq K$, \\
	ii) the Rx-scaling factor $a\in \mathcal{S}_i$, and\\
	iii) the Tx-scaling factor $b_k$ is given by \eqref{optimal_b}, $\forall k\in \{1,\cdots,K\}$. 
\end{definition}
\begin{remark}
	Different from the optimal policy, where the critical number depends on all the values of the channel coefficients, a first-$\imath$ policy 
	simply determines its critical number based on the total number of sensors $K$. The first $i=\imath(K)$ sensors with the smallest channel coefficients use the maximum power for transmission.
\end{remark}


\subsubsection{Scaling Law of Average Computation MSE}
For a first-$\imath$ policy, using Definition~\ref{def:policy}, it is clear that
\begin{equation}\label{a_inequal}
\frac{1}{h_{\imath(K)+1} \sqrt{P}} <  a \leq \frac{1}{h_{\imath(K)} \sqrt{P}}.
\end{equation}
Taking the inequality \eqref{a_inequal} into \eqref{target_temp}, an upper bound of $\mse$ is obtained as
\begin{equation} \label{upper_bound}
\begin{aligned}
\mse &\leq  \sum_{k=1}^{\imath(K)} \left(\frac{h_k}{h_{\imath(K)+1}} -1\right)^2
+
\frac{\sigma^2 }{P} \frac{1}{h^2_{\imath(K)}}
\leq
\imath(K) + \frac{\sigma^2 }{P} \frac{1}{U_{\imath(K)}}.
\end{aligned}
\end{equation} 
and a lower bound of $\mse$ is obtained as 
\begin{equation} \label{low_bound}
\begin{aligned}
\mathsf{MSE}&>\sum_{k=1}^{\imath(K)} \left(\frac{h_k}{h_{\imath(K)}} -1\right)^2
+
\frac{\sigma^2 }{P} \frac{1}{U_{\imath(K)+1}}.
\end{aligned}
\end{equation} 
Then, we have the following result.
\begin{theorem}\label{theory:first_i}
	\normalfont
For a first-$\imath$ policy, the average computation MSE of the AirComp system has the following properties that
\begin{enumerate}
	\item if $\liminf_{K\rightarrow \infty} \imath(K)\leq 2$,
\begin{equation} \label{key_scaling2}
\limsup\limits_{K\rightarrow \infty}\frac{\myexpectation{\mse}}{K} \geq \frac{\sigma^2}{3 P},
\end{equation}
\item if $\liminf_{K\rightarrow \infty} \imath(K)> 2$ and $\limsup_{K\rightarrow \infty} \imath(K)/K = c >0$, 
\begin{equation} \label{key_scaling1}
\limsup\limits_{K\rightarrow \infty}\frac{\myexpectation{\mse}}{K} \geq \frac{1}{c}\mu\left(\frac{c}{1+c}\right),
\end{equation}
where 
\begin{equation}\label{haha}
\mu(x) \triangleq \frac{x}{1-x}-\log (1-x)-2 \sqrt{\frac{x}{1-x}} \sin ^{-1}\left(\sqrt{x}\right),
\end{equation}
\item if $\liminf_{K\rightarrow \infty} \imath(K)> 2$ and $\lim_{K\rightarrow \infty} \imath(K)/K = 0$, 
%
\begin{equation} \label{key_scaling}
\frac{\myexpectation{\mse}}{K} \sim \frac{\imath(K)}{K}+\frac{1}{\imath(K)}, K\rightarrow \infty.
\end{equation}
\end{enumerate}

\end{theorem}
\begin{proof}
	See Appendix B.
\end{proof}
For case 1) in \eqref{key_scaling2}, the average computation MSE does not converge to zero and such the first-$\imath$ policy is not computation-effective.
From case 2) in \eqref{key_scaling1}, interestingly, we see that the energy-greedy policy in Definition~\ref{def:energy_greedy}, which has $\limsup_{K\rightarrow \infty} \imath(K)/K = 1$, is not computation-effective, since $\limsup_{K\rightarrow \infty} \myexpectation{\mse}/K \geq \mu(1/2)\approx 0.12 >0$.
For case 3), from \eqref{key_scaling}, ${\myexpectation{\mse}}/{K}$ converges to zero iff $\imath(K)=o(K)$ and $\liminf_{K\rightarrow \infty}\imath(K) \rightarrow \infty$. Moreover, the average computation MSE has the decay rate of $1/\sqrt{K}$ when $\imath(K) \sim \sqrt{K}$, 
while if $ \imath(K) = o(\sqrt{K})$ or $ 1/\imath(K) = o(1/\sqrt{K})$, we have $1/\imath(K)\gg \imath(K)/K$ or $1/\imath(K)\ll \imath(K)/K$ in \eqref{key_scaling}, respectively, and the decay rate of the average computation MSE is larger than $1/\sqrt{K}$. Therefore, we have the following result.
%
\begin{proposition}\label{scaling-mse}
	\normalfont
	A first-$\imath$ policy is computation-effective iff	
	$\imath(K)=o(K)$ and $\liminf_{K\rightarrow \infty}\imath(K) \rightarrow\infty$.	
	The largest decay rate of the average computation MSE achieved by first-$\imath$ policies~is
	\begin{equation}
		\frac{\myexpectation{\mse}}{K} \sim \frac{1}{\sqrt{K}}, K\rightarrow \infty,
	\end{equation}
when the critical-number function $\imath(K) \sim \sqrt{K}$.
\end{proposition}

From \eqref{mse_optimal}, the computation-optimal policy surely results in an average computation MSE no larger than that of a first-$\imath$ policy, we have the following result.
\begin{proposition}\label{scaling-mse:2}
	\normalfont
	The computation-optimal policy is a computation-effective one, and its average computation MSE has at least a decay rate of $O({1}/{\sqrt{K}})$ when $K\rightarrow \infty$.
\end{proposition}

\subsubsection{Scaling Law of Average Power Consumption}
From Definition~\ref{def:policy}, the power consumption of a first-$\imath$ policy is
\begin{equation} \label{Power}
\pw = \imath(K) P + \frac{1}{a^2} \sum_{k=\imath(K) +1}^{K}  \frac{1}{U_k},\ a\in\mathcal{S}_{\imath(K)}.
\end{equation}
We have the following result of its average power consumption.
\begin{theorem} \label{theory:power}
	\normalfont
For a first-$\imath$ policy, the average power consumption of the AirComp system has the following properties that
\begin{enumerate}
	\item if $\liminf_{K\rightarrow \infty} \imath(K)/K =1$,
	\begin{equation} \label{kkey_scaling2}
	\limsup\limits_{K\rightarrow \infty}\frac{\myexpectation{\pw}}{K} = P,
	\end{equation}
	\item if $\limsup_{K\rightarrow \infty} \imath(K)/K = c' >0$, 
	\begin{equation} \label{kkey_scaling1}
	\limsup\limits_{K\rightarrow \infty}\frac{\myexpectation{\pw}}{K} \geq c' P,
	\end{equation}
	\item if $\lim_{K\rightarrow \infty} \imath(K)/K = 0$, 	
	\begin{equation} \label{kkey_scaling}
	O\left(\frac{\imath(K)}{K}\right) \leq  \limsup\limits_{K\rightarrow \infty} \frac{\myexpectation{\mathsf{PW}}}{K} \leq O\left(\frac{\imath(K)\log(K)}{K}\right).
	\end{equation}
\end{enumerate}	
\end{theorem}
\begin{proof}
	See Appendix C.
\end{proof}

From \eqref{kkey_scaling}, a first-$\imath$ policy is energy-efficient as long as the critical number function $\imath(K)$ has a lower divergence rate of $K/\log(K)$, and the decay rate of the average power consumption is no larger than $O\left(\imath(K)/K\right)$.
Using Proposition~\ref{scaling-mse} and Theorem~\ref{theory:power}, the following result can be obtained directly.
\begin{proposition} \label{prop:sqrtK}
	\normalfont	
The computation-effective first-$\imath$ policy achieving the minimum average computation MSE, i.e., $\imath(K)\sim \sqrt{K}$ when $K\rightarrow\infty$, is also energy-efficient, and its average power consumption has a decay rate between $O\left(1/\sqrt{K}\right)$ and $O\left(\log(K)/\sqrt{K}\right)$.
\end{proposition}

Note that the scaling-law results of first-$\imath$ policies for average power consumption cannot provide insights directly into that of the computation-optimal policy. This is because the power consumption of the computation-optimal policy in \eqref{power} relies on the optimal critical number, which is determined by the sequence $\{g_i\}$ in \eqref{g_i}, and the statistics of the optimal critical number is difficult to analyze. 
Nevertheless, we numerically show that this policy is also an energy-efficient one in the following section.
Therefore, by using the results obtained in this section, the computation effectiveness and energy efficiency of benchmark policies 1 and 2, computation-optimal policy, and two first-$\imath$ policies, i.e., $\imath(K)=\max\{1,\lfloor \sqrt{K} \rfloor\}$ and $\imath(K)=\max\{1,\lfloor {K}/2 \rfloor\}$ are summarized in Table~\ref{tab:1}. 
\begin{table*}[t]
	\centering
	\renewcommand{\captionfont}{\small} \renewcommand{\captionlabelfont}{\small}		
	\caption{Computation effectiveness and energy efficiency of AirComp Policies.}
	\begin{tabular}{|l|c|c|}
		\hline
		& Computation-Effective Policy & Energy-Efficient Policy\\ \hline\hline
		Benchmark Policy 1~\cite{GuangxuMIMO,Xiaoyang} & \xmark  &  \cmark \\ \hline
		Benchmark Policy 2 & \xmark & \xmark\\ \hline
		Computation-Optimal Policy & \cmark & \cmark\\ \hline
		First-$\imath$ Policy with $\imath(K)=\sqrt{K}$ & \cmark  & \cmark \\ \hline
		First-$\imath$ Policy with $\imath(K)={K}/2$ & \xmark & \xmark \\ \hline
	\end{tabular}
	\label{tab:1}
	\vspace{-0.8cm}
\end{table*}

\vspace{-0.5cm}
\section{Numerical Results} \label{sec:num}
In this section, we numerically evaluate the average computation MSE and the average power consumption of the five Tx-Rx scaling policies of the AirComp system in Table~\ref{tab:1}.
{Recall that benchmark policy 1 is the channel-inversion policy which has been commonly considered in the literature of AirComp~\cite{GuangxuMIMO,Xiaoyang}, and benchmark policy 2 is the energy-greedy policy defined in Sec.~\ref{sec:benchmark}.}
Note that for the first-$\imath$ policies, the Rx-scaling factor is chosen as $a = \left(1/h_{\imath(K)+1}+1/h_{\imath(K)}\right)/(2\sqrt{P}) \in \mathcal{S}_{\imath(K)}$.
{Unless otherwise stated, we set the number of sensors $K=10$, the sensor's transmission power limit $P=10$, the receiving noise power $\sigma^2=1$, and the channel coefficients between the sensors and the receiver follow the i.i.d. Rayleigh distribution with unit variance, i.e., if a sensor transmits signal with peak power, the average received SNR is $10$~dB.}
The average computation MSE, $\myexpectation{\mse}/K$, and the average power consumption, $\myexpectation{\pw}/K$, induced by different policies are evaluated by using Monte Carlo simulation with $10^6$ random channel realizations for calculating the average of ${\mse}/K$ and ${\pw}/K$ based on \eqref{MSE} and \eqref{def_pw}, respectively.
Also, the standard deviations of $\mse/K$ and $\pw/K$ are evaluated as the confidence intervals of the average computation MSE and the average power consumption, respectively.
For the computation-optimal policy, we also evaluate the average and the standard deviation of its critical number using Monte Carlo simulation with $10^6$ points.

In Fig.~\ref{fig:i}, using Theorem~\ref{theory:optimal}, we plot the average critical number of the computation-optimal policy, $\myexpectation{i^\star}$, and the confidence region of $i^\star$ with different number of sensors $K$.
We see both the average and the standard deviation of the  critical number monotonically increase with $K$.
Also, we plot the critical number of two first-$\imath$ policies with $\imath(K) =\sqrt{K}$ and ${K}^{1/3}$, respectively.
It shows that the scaling law of the average critical number of the computation-optimal policy has the properties that $\myexpectation{i^\star}>{K}^{1/3}$ and $\myexpectation{i^\star} < \sqrt{K}$, when $K>10$.

\begin{figure}[t]
	\renewcommand{\captionfont}{\small} \renewcommand{\captionlabelfont}{\small}	
	\centering
	\includegraphics[scale=0.6]{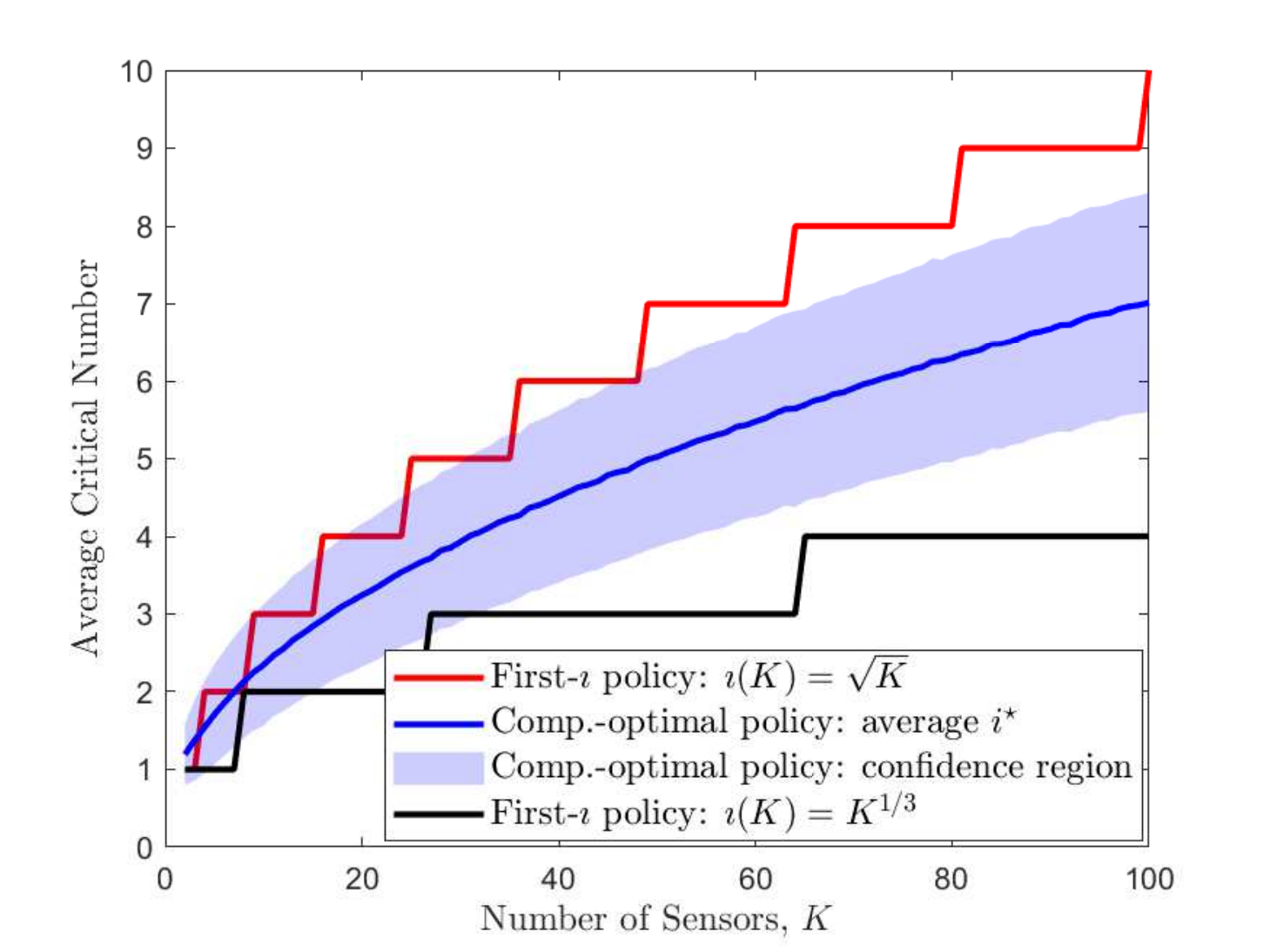}
	\vspace{-0.5cm}
	\caption{The average critical number of the computation-optimal policy versus the number of sensors.}
	\label{fig:i}
	\vspace{-0.5cm}	
\end{figure}

\begin{figure*}[t]
	\renewcommand{\captionfont}{\small} \renewcommand{\captionlabelfont}{\small}
	\minipage{0.5\textwidth}
	\centering
	\includegraphics[scale=0.6]{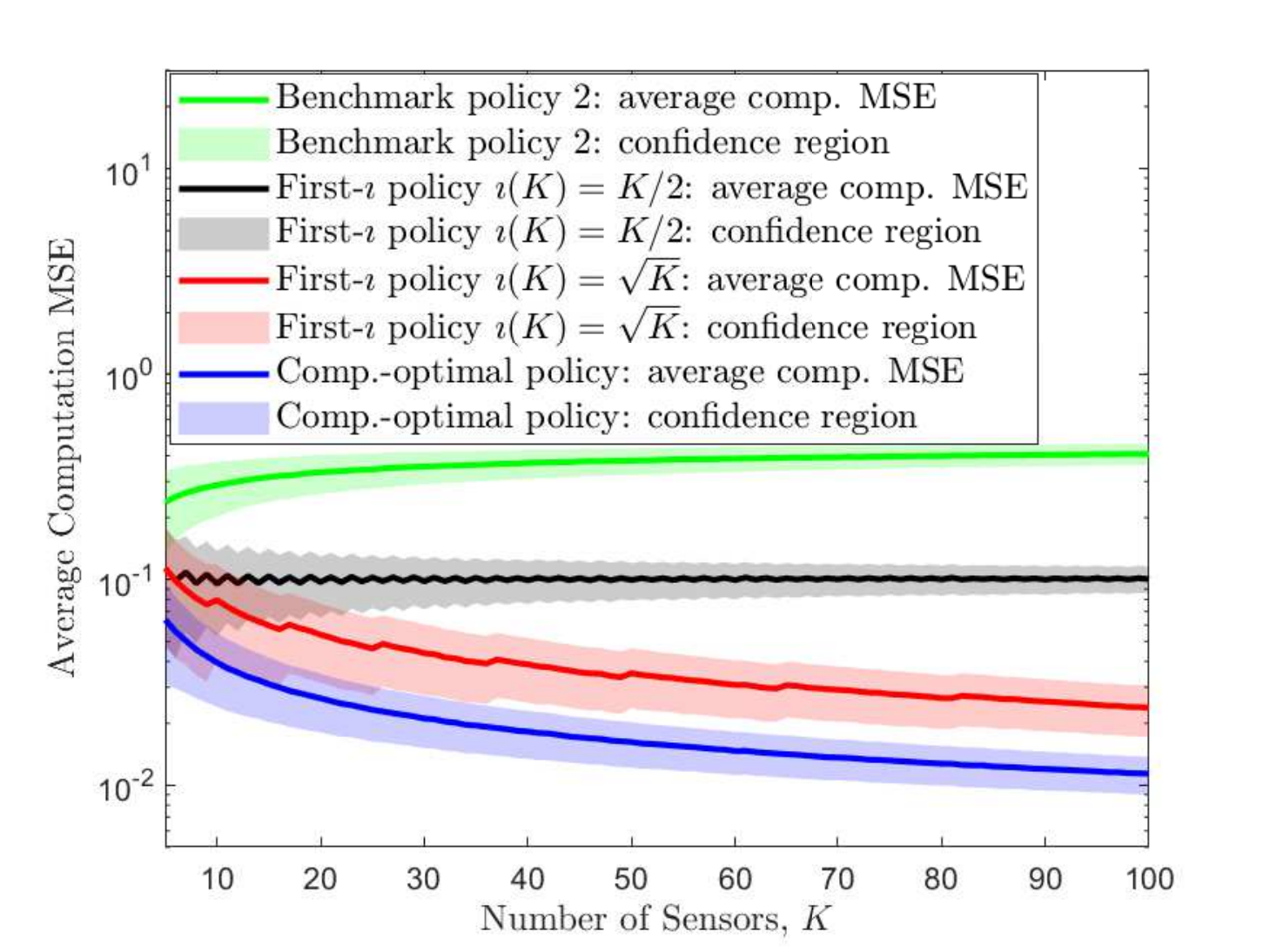}
	\vspace{-1.2cm}
	\caption{The average computation MSE versus $K$.}
	\vspace{-0.5cm}
	\label{fig:MSE}
	\endminipage
	\minipage{0.5\textwidth}
	\centering
	\includegraphics[scale=0.6]{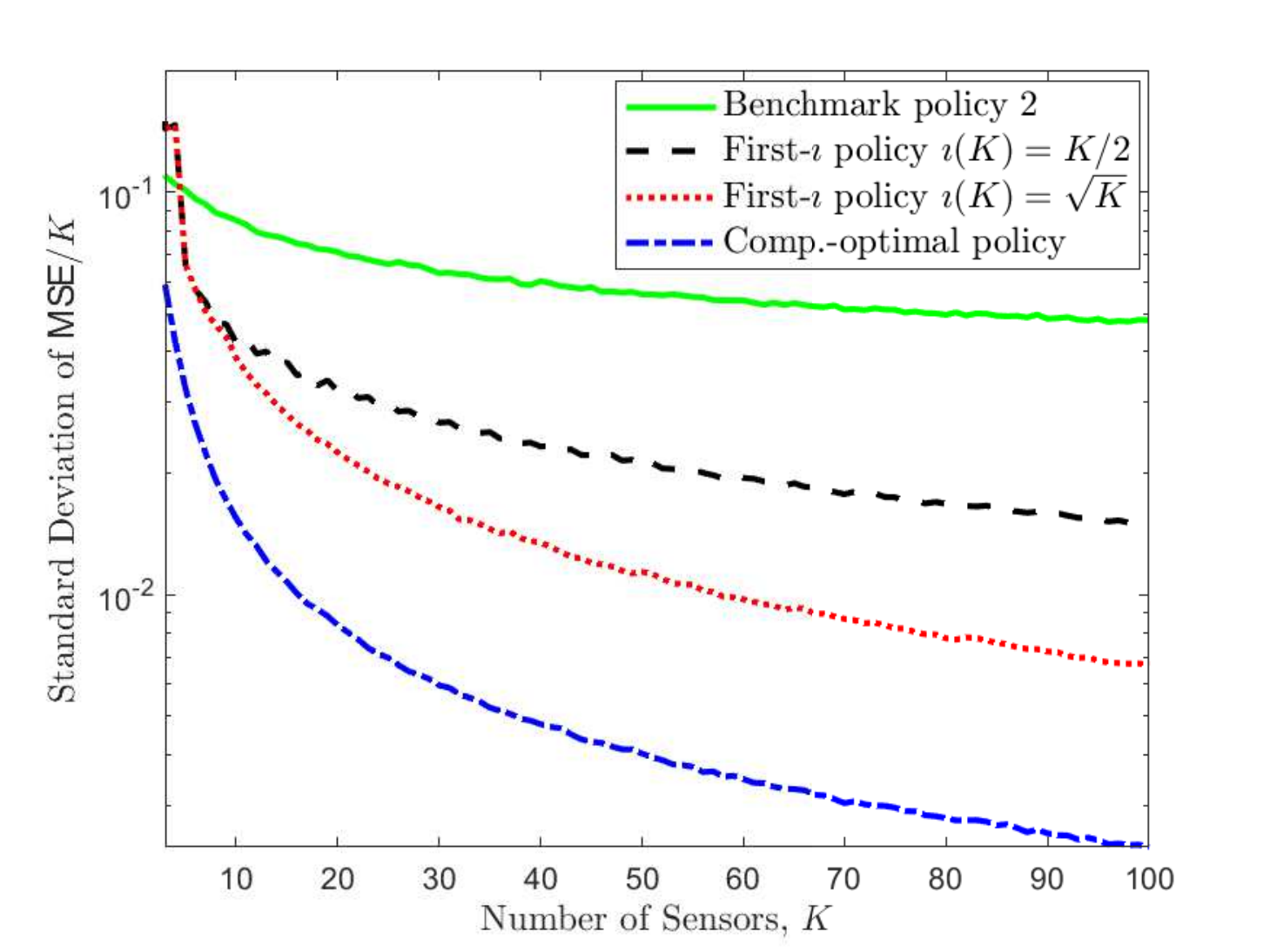}
	\vspace{-1.2cm}
	\caption{The standard deviation of $\mse/K$ versus $K$.}
	\vspace{-0.5cm}
	\label{fig:MSE_deviation}
	\endminipage
	\vspace*{-0.5cm}
\end{figure*}

{In Figs.~\ref{fig:MSE} and~\ref{fig:MSE_deviation}, we plot the average computation MSE of the AirComp system and the standard deviation of $\mse/K$, respectively, with different policies in Table~\ref{tab:1}, excluding benchmark policy 1, which has an infinite average computation MSE.}
We see that the computation-optimal policy has a remarkably lower average computation MSE than the other policies, and the  policy with $\imath(K)=\sqrt{K}$ is better than the one with $\imath(K)={K}/2$, which is better than benchmark policy 2.
Also, it can be observed that both the policy with $\imath(K)=\sqrt{K}$ and the computation-optimal policy have average computation MSEs approaching to zero with the increasing $K$, which verifies Propositions~\ref{scaling-mse} and \ref{scaling-mse:2}, respectively. However, benchmark policy 2 and the policy with $\imath(K)=K/2$ have average computation MSEs bounded away from zero and converge to $0.35$ and $0.15$, respectively, which are in line with Corollary~\ref{cor:benchmark2-1} and Theorem~\ref{theory:first_i}.
Nevertheless, all these four policies have diminishing standard deviations of $\mse/K$ with the increasing $K$, which means that $\mse/K$ convergences to average computation MSE in probability when $K\rightarrow \infty$. Also, it is interesting to see that the policy with a lower average computation MSE has a smaller standard deviation when $K>10$.

%
%
%
%
%

\begin{figure*}[t]
	\renewcommand{\captionfont}{\small} \renewcommand{\captionlabelfont}{\small}
	\minipage{0.5\textwidth}
	\centering
	\includegraphics[scale=0.6]{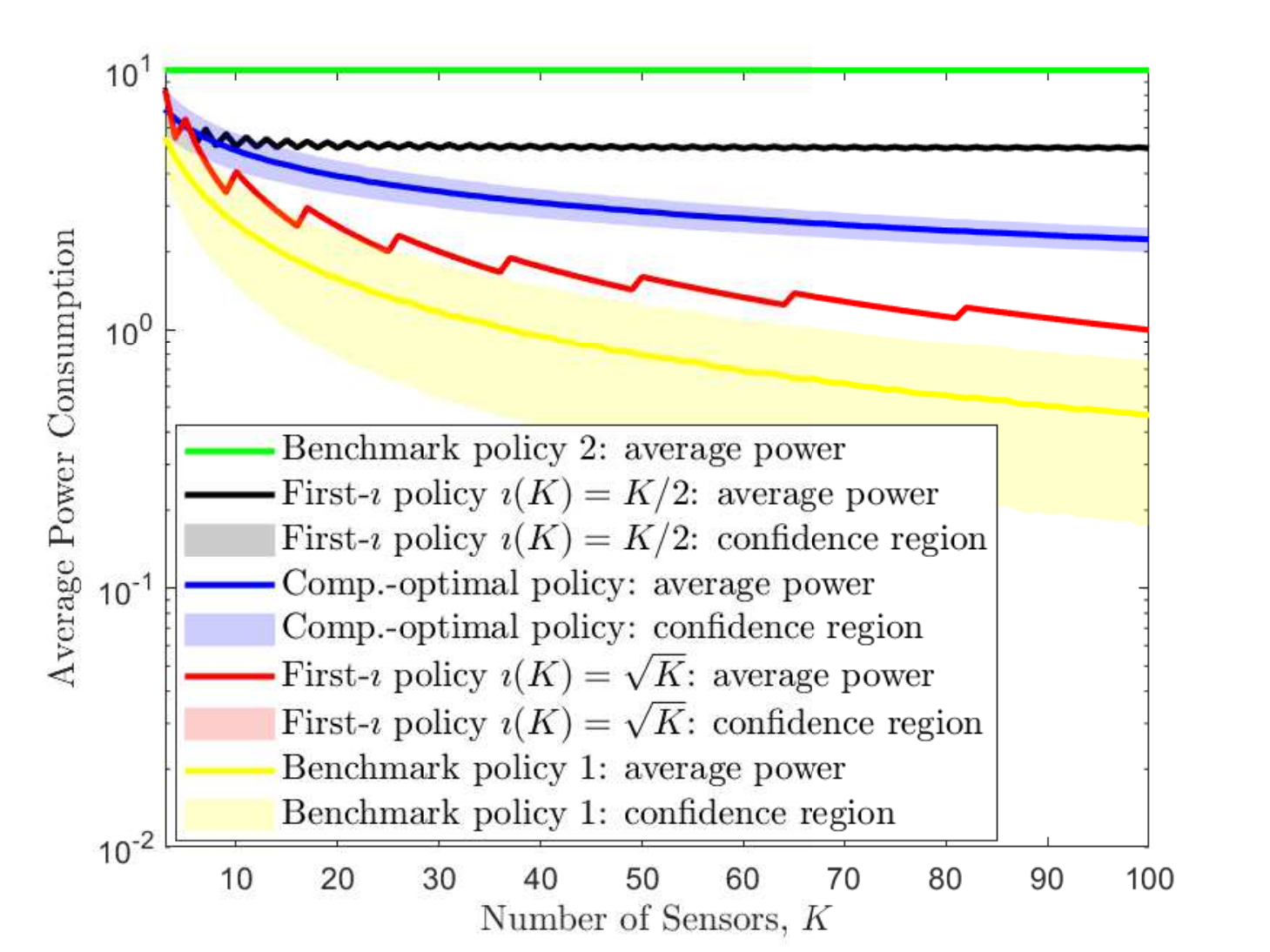}
	\vspace{-1.1cm}
	\caption{The average power consumption versus  $K$.}
	\vspace{0.6cm}
	\label{fig:power}
	\endminipage
	\minipage{0.5\textwidth}
	\centering
	\includegraphics[scale=0.6]{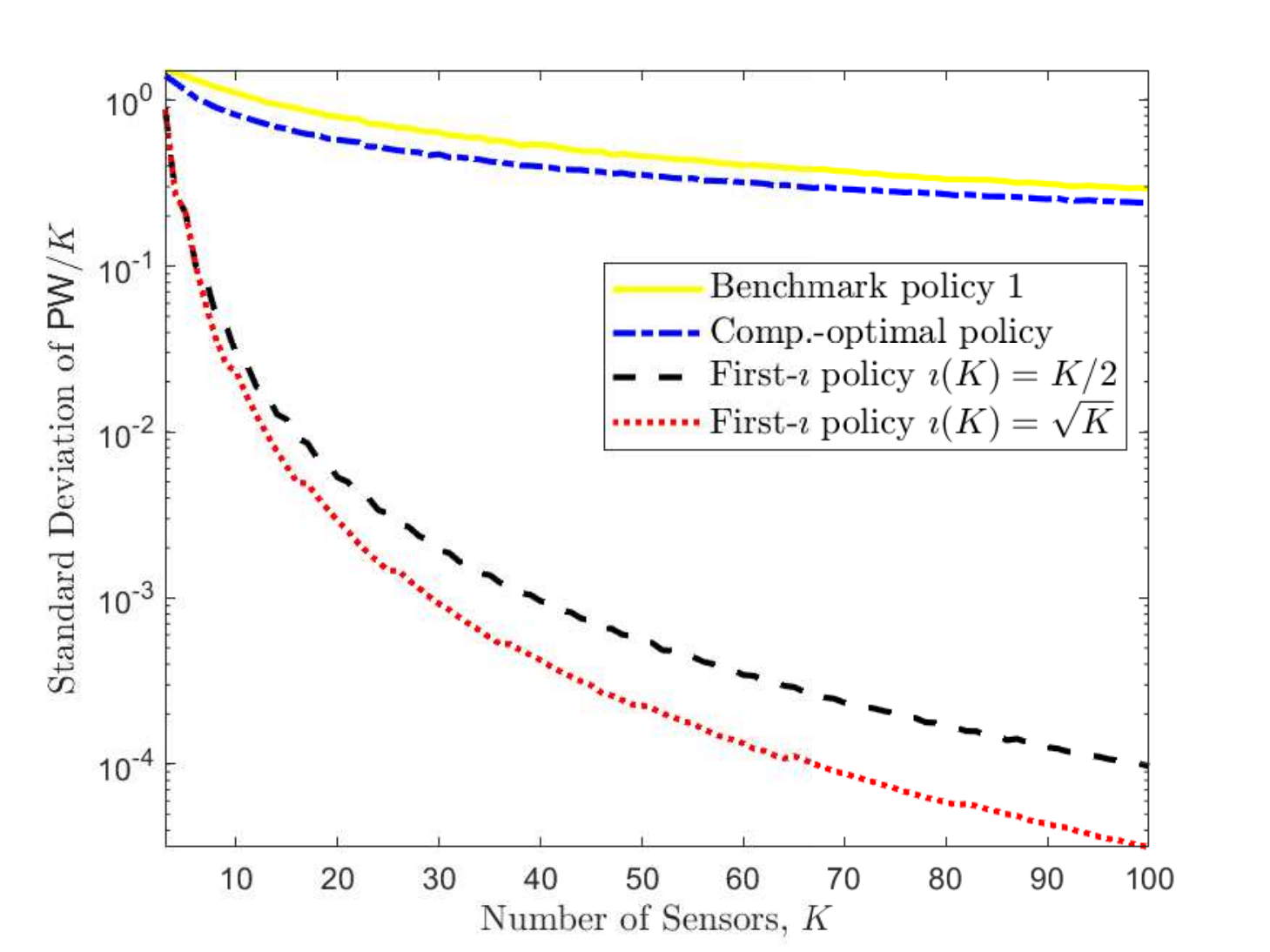}
	\vspace{-1.1cm}
	\caption{The standard deviation of ${\pw}/K$ versus $K$, where benchmark policy 2 which has zero standard deviation of ${\pw}/K$, is not included in the logarithmic-scale plot.}
	\vspace{-0.5cm}
	\label{fig:power_deviation}
	\endminipage
	\vspace*{-0.5cm}
\end{figure*}

In Figs.~\ref{fig:power} and~\ref{fig:power_deviation}, we plot the average power consumption of the AirComp system, $\myexpectation{\pw}/K$ and the standard deviation of $\pw/K$, respectively, with different policies in Table~\ref{tab:1}.
{Note that the confidence regions (standard derivations) of the first-$\imath$ policies (i.e., the red and black lines) in Fig.~\ref{fig:power} is too narrow to be visible, and the standard derivations have been clearly illustrated as the red and black dashed lines in Fig.~\ref{fig:power_deviation}.}
We see that benchmark policies 1 and 2 have the lowest and the highest power consumption, respectively.
The policy with $\imath(K)=\sqrt{K}$ and the computation-optimal policy both have average power consumption approaching to zero with the increasing $K$, which is in line with Propositions~\ref{prop:sqrtK}, and the former has a lower power consumption than the latter.
Comparing Fig.~\ref{fig:power} with Fig.~\ref{fig:MSE}, the computation-optimal policy has a better computation performance but a higher power consumption than the policy with $\imath(K)=\sqrt{K}$, which again shows the design tradeoff between computation effectiveness and energy efficiency.
Also, we see the average power consumption of the policy with $\imath(K)=K/2$ converges to $5$, which is greater than $P/3=3.3$, which is in line with Theorem~\ref{theory:power}.

From Fig.~\ref{fig:power_deviation}, it can be observed that all the polices have diminished standard deviations of $\pw/K$ with the increasing $K$, which means that $\pw/K$ convergences to average power consumption in probability when $K\rightarrow \infty$.
It is interesting to see that the standard deviations of the first-$\imath$ policies are much smaller than that of the computation-optimal policy. 
This is mainly because the critical number is deterministic for the former and is stochastic for the latter. Recall that for a first-$\imath$ policy, the critical number does not rely on the random channel realizations, while the critical number of the computation-optimal policy heavily relies on the channel realizations, and thus has a large variance as shown in Fig.~\ref{fig:i}.

\begin{figure*}[t]
	\renewcommand{\captionfont}{\small} \renewcommand{\captionlabelfont}{\small}
	\minipage{0.5\textwidth}
	\centering
	\includegraphics[scale=0.6]{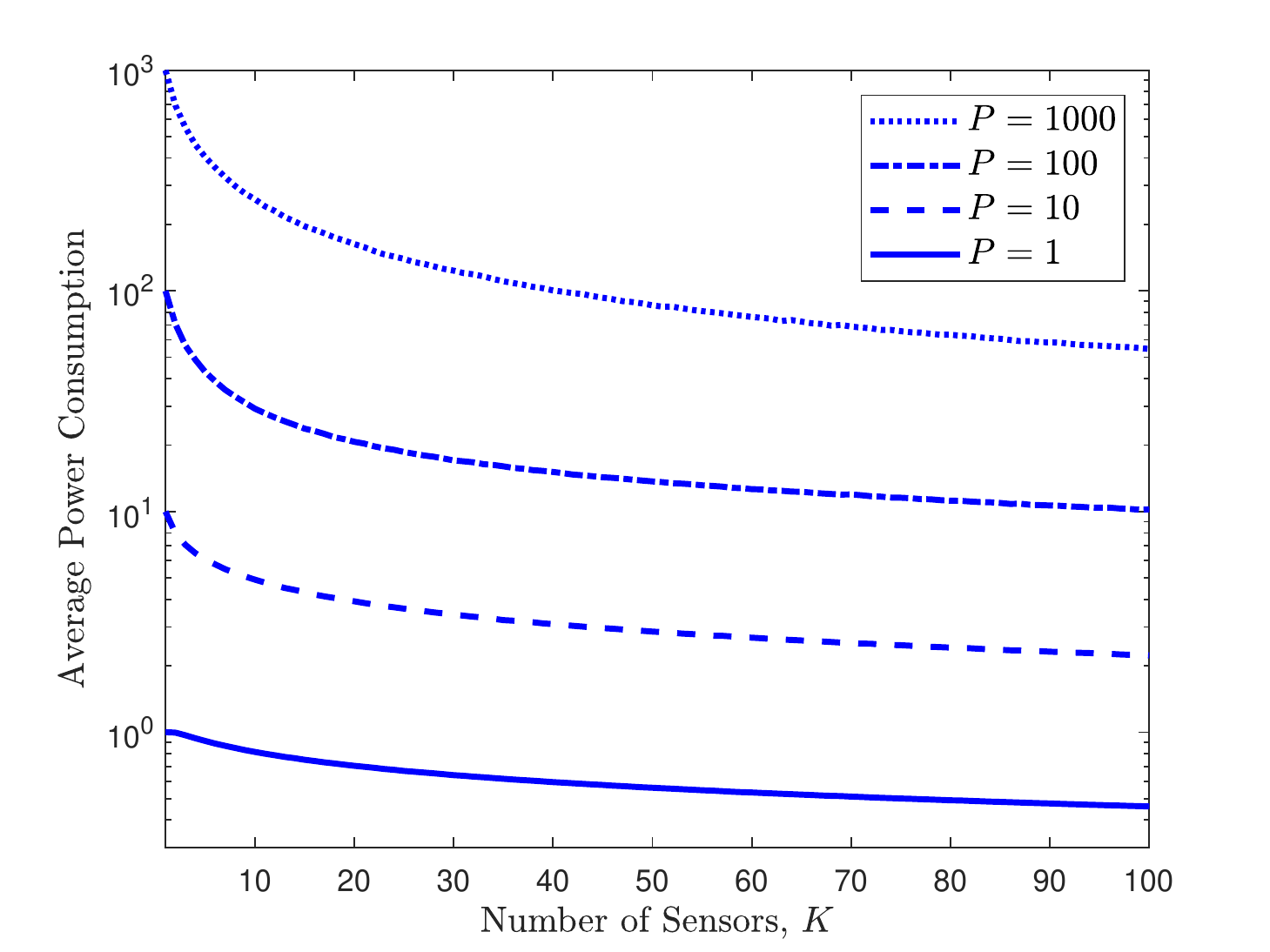}
	\vspace{-1.2cm}
	\caption{The average critical number of the computation-optimal policy versus the number of sensors.}
	\label{fig:compare}
	\endminipage
	\hspace{0.1cm}
	\minipage{0.5\textwidth}
	\centering
	\includegraphics[scale=0.6]{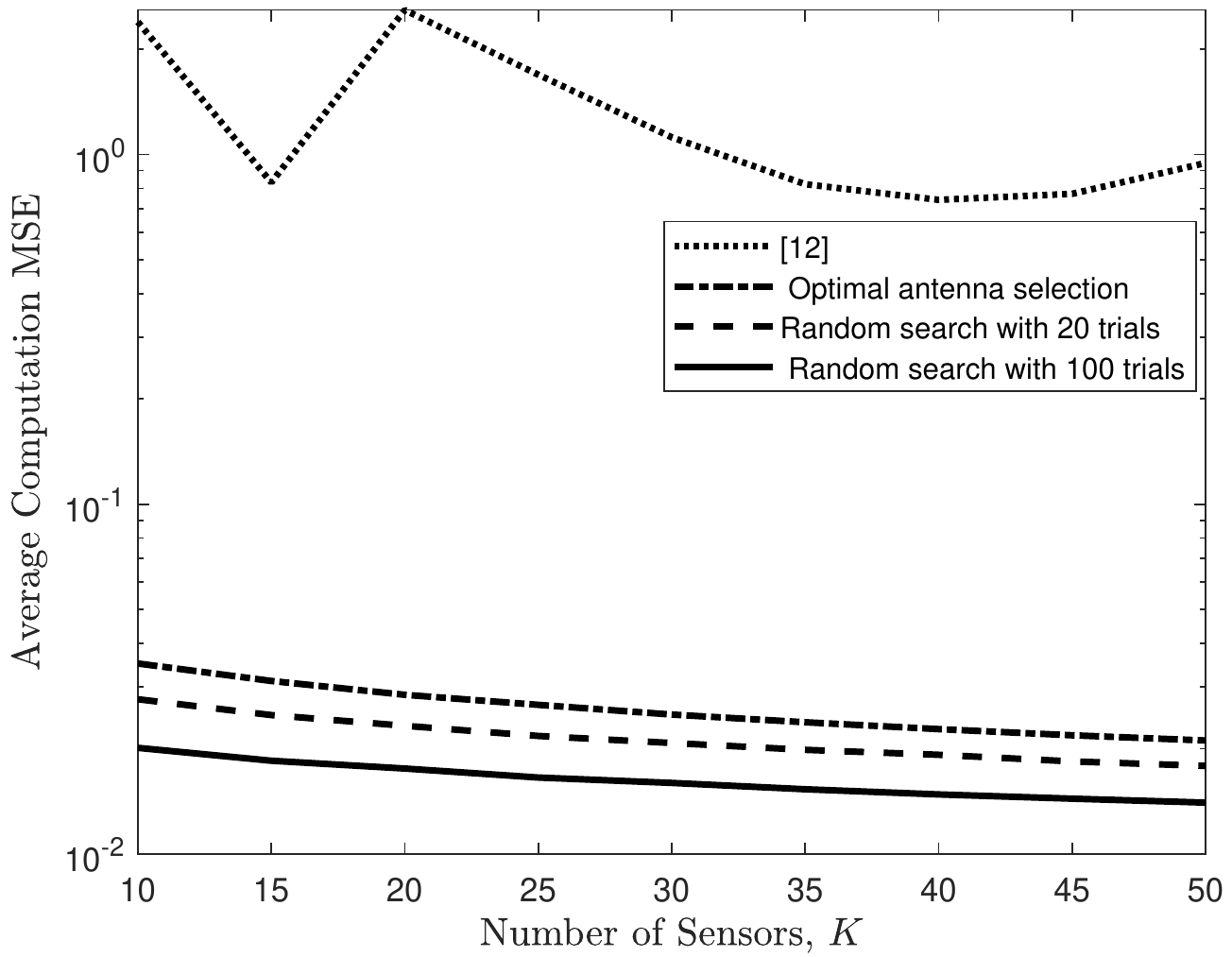}
	\vspace{-0.2cm}
	\caption{Multi-antenna receiver case: the average power consumption versus  $K$, where $N=8$.}
	\label{fig:last}
	\endminipage
	\vspace*{-0.9cm}
\end{figure*}

Fig.~\ref{fig:compare} shows the average power consumption of the computation-optimal policy with different transmission-power limits. We see that the average power consumption in all different cases decays to zero with the increasing number of sensors, which verifies that the policy is an energy-efficient one. 


{In Fig.~\ref{fig:last}, we evaluate the average computation MSE achieved in the multi-antenna receiver case with different policies, including the policy given in~\cite{GuangxuMIMO}, the optimal antenna-selection policy and the random search of unit vector policy given in Sec.~\ref{sec:multi}. The number of receiver antenna is $N=8$.
We see that the proposed optimal antenna-selection policy can achieve at least a $40$ times lower $\mse$ than that of the method in~\cite{GuangxuMIMO}.
As expected, the MSE reduces with the increasing random search trials.
It is important to see that similar to the single-antenna case, the the average computation MSEs achieved by the optimal antenna selection and the random search policies monotonically decrease with the increasing number of sensors in the multi-antenna receiver case.
The fluctuations of the policy~\cite{GuangxuMIMO} imply that the average computation MSE of the policy does not exist in the Rayleigh fading scenario.}


\vspace{-0.5cm}
\section{Concluding Remarks}
In this work, we have derived the computation-optimal policy of the AirComp system, compared the AirComp system with the traditional MAC system, and investigated the ergodic performance of the AirComp system with different Tx-Rx scaling policies in terms of the number of sensors.
Our results have shown that the computation-optimal policy has a vanishing average computation MSE and a vanishing average power consumption with the increasing number of sensors.
By comparing the performance of the computation-optimal policy with that of the proposed first-$\imath$ policies, it reveals a design tradeoff between computation effectiveness and energy efficiency, which is very important for AirComp-system implementation with practical constraints of computation accuracy and energy consumption. {Inspired by such a tradeoff in AirComp, we will investigate energy-efficiency optimization problems in our future work.}

\setcounter{equation}{0}
\renewcommand\theequation{A.\arabic{equation}}
\vspace{-0.3cm}
\section*{Appendix A: Proof of Corollary~\ref{cory:2}}
The joint distribution function $f(u_1,u_2,\cdots,u_K )$ can be rewritten as
\begin{equation}\label{app:1}
f(u_1,u_2,\cdots,u_K )
=f(u_2,\cdots,u_K \vert u_1) f(u_1).
\end{equation}
Since the average power in \eqref{ave_power_1} does not rely on the order of the largest $K-1$ channel power gains, $U_k, k>1$, in the rest of the proof, we treat $\{U_k\}$ as $K$ independent exponential random variables that $U_k, \forall k>1$, is no smaller than $U_1$.
In this sense, the conditional joint distribution in \eqref{app:1} can be rewritten as
\vspace{-0.5cm}
\begin{equation}
\begin{aligned}
&f(u_2,\cdots,u_K \vert u_1, \Xi_1)
\!= f(u_2 \vert u_1, \Xi_1)\cdots f(u_K \vert u_1, \Xi_1)
\!=\!\left\lbrace
\begin{aligned}
&e^{(K-1)u_1 - (u_2+\cdots+u_K)},&&\!\!u_2,\cdots,u_K>u_1\\
&0, &&\!\!\text{else},
\end{aligned}
\right.
\end{aligned}
\end{equation}
where $\Xi_1$ is the event that $U_1\leq U_2,\cdots,U_K$, and $f(u_1)$ given in \eqref{pdf_1} can be denoted as $f(u_1\vert \Xi_1)$. 

Then, it can be obtained that 
\begin{align}\label{temp_result}
&\myexpectation{\sum_{k=2}^{K} \frac{U_1}{U_k} \bigg\vert \Xi_1}
=\frac{K \ln K - (K-1)}{K-1}.
\end{align}

Corollary~\ref{cory:2} is obtained by taking \eqref{temp_result} into \eqref{ave_power_1}.

\section*{Appendix B: Proof of Theorem~\ref{theory:first_i}}
From \eqref{upper_bound}, an upper bound of the average computation MSE is derived~as
\begin{align}\label{upp_upp}
\frac{\myexpectation{\mse}}{K} 
&\leq \frac{\imath(K)}{K}+ \frac{\sigma^2 }{P K} \myexpectation{\frac{1}{U_{\imath(K)}}}
= \frac{\imath(K)}{K}+ \frac{\sigma^2 }{P K} \myexpectation{\frac{1}{\sum_{j=1}^{\imath(K)} \frac{Z_j}{K-j+1}}}\\\label{inverse_gamma}
&\leq \frac{\imath(K)}{K}+ \frac{\sigma^2 }{P K} \myexpectation{\frac{K}{\sum_{j=1}^{\imath(K)} {Z_j}}}
= \frac{\imath(K)}{K}+ \frac{\sigma^2 }{P} \frac{1}{\imath(K)-1},
\end{align}
where the equality in \eqref{upp_upp} is due to the property that for $K$ random samples from an exponential distribution with parameter $1$, the order statistics $U_i$ for $i = 1,2,3, \cdots, K$ each has the distribution
$
U_i \stackrel{d}{=} \sum_{j=1}^{i} \frac{Z_j}{K-j+1}
$~\cite{orderstatistics},
and $Z_j,j=1,\cdots,i$ are i.i.d. standard exponential random variables.
The equality in \eqref{inverse_gamma} is due to the fact that  ${\frac{1}{\sum_{j=1}^{\imath(K)} {Z_j}}}$ follows inverse gamma distribution with mean $\frac{1}{\imath(K)-1}$.

From \eqref{low_bound}, a lower bound of the average computation MSE is derived~as
\begin{align}\label{low_low}
\frac{\myexpectation{\mse}}{K} 
&\geq \frac{1}{K}\sum_{k=1}^{\imath(K)} \myexpectation{ \left(\frac{h_k}{h_{\imath(K)}} -1\right)^2}
+
\myexpectation{\frac{\sigma^2 }{P K} \frac{1}{U_{\imath(K)+1}}}.
\end{align}
Then, we derive the lower bounds of the first and second terms on the right-hand side (RHS) of \eqref{low_low} using the following technical lemma.
\setcounter{section}{3}
\setcounter{lemma}{0}
\begin{lemma}\label{lem:ratio}
	\normalfont
	Let $X_1,X_2,\cdots,X_K$ be a random sample from the standard exponential distribution, and let $X_{(1)},X_{(2)},\cdots,X_{(K)}$ denote
	the order statistics obtained from this sample. The expectation of the ratio $\frac{X_{(i)}}{X_{(j)}}$ has the inequality
	$
	\myexpectation{\frac{X_{(i)}}{X_{(j)}}} 
	< \frac{(i+1)}{(K-i+1)}\frac{K}{(j-2)}, \forall j>2.
	$
\end{lemma}
\setcounter{section}{2}
\setcounter{lemma}{0}
\begin{proof}
It can be derived that	
	\begin{align} \label{Holder}
	\myexpectation{\frac{X_{(i)}}{X_{(j)}}} 
	& \leq \sqrt{\myexpectation{X_{(i)}} \myexpectation{\frac{1}{X_{(j)}}} }
	\leq \sqrt{\frac{(i^2+i)}{(K-i+1)^2}\frac{K^2}{(j-1)(j-2)}} 
	< \frac{(i+1)}{(K-i+1)}\frac{K}{(j-2)}, \forall j>2,
	\end{align}
	where the first inequality in \eqref{Holder} is due to the Cauchy-Schwarz inequality, i.e., $\myexpectation{\vert X Y \vert}\leq \sqrt{\myexpectation{ X^2}\myexpectation{ Y^2}}$,
	and the second inequality in \eqref{Holder} is obtained by using
	the inequalities
	\begin{equation}\label{exp2}
	\frac{1}{K}\sum_{j=1}^{i} {Z_j}\preceq X_i \stackrel{d}{=} \sum_{j=1}^{i} \frac{Z_j}{K-j+1} \preceq \frac{1}{K-i+1}\sum_{j=1}^{i} {Z_j}, \forall i,
	\end{equation}
	and the property that $\sum_{j=1}^{i}Z_j$ and $\frac{1}{\sum_{j=1}^{i}Z_j}$ follows the gamma distribution $\mathrm{Gamma}(i,1)$ and the inverse gamma distribution $\mathrm{Inv-Gamma}(i,1)$, respectively.
\end{proof}

For the second term on the RHS of \eqref{low_low}, using the Jensen's inequality and \eqref{exp2}, we have
\begin{equation} \label{second}
\myexpectation{\frac{\sigma^2 }{P K} \frac{1}{U_{\imath(K)+1}}} \geq \frac{\sigma^2 }{P K} \frac{K-\imath(K)}{\imath(K)+1}.
\end{equation} 
Thus, if $\liminf\limits_{K\rightarrow\infty} \imath(K)\leq 2$, $\limsup\limits_{K\rightarrow\infty} \frac{\myexpectation{\mse}}{K} \geq \frac{\sigma^2 }{P K} \frac{K-\liminf\limits_{K\rightarrow\infty} \imath(K)}{\liminf\limits_{K\rightarrow\infty} \imath(K)+1}\geq  \frac{\sigma^2 }{3 P}, K\rightarrow \infty$, completing the proof of \eqref{key_scaling2}. In the following, we assume that $\liminf\limits_{K\rightarrow\infty} \imath(K) > 2$.

For the first term on the RHS of \eqref{low_low}, we have the following inequality
\begin{align} \label{jensen}
&\myexpectation{\! \left(\!\frac{h_k}{h_{\imath(K)}} -1\!\right)^2\!}
\!\geq\! \left(\!1-\sqrt{\myexpectation{\frac{U_k}{U_{\imath(K)}}}}\!\right)^2 
\!\geq\!
\left(\!\left[\!1-\sqrt{\frac{K}{K-k+1} \frac{k+1}{\imath(K)-2}}\!\right]^+\!\right)^2\!,\forall \imath(K)\!>\!2,k\!\leq \!\imath(K).
\end{align}
where the first inequality is due to the Jensen's inequality and the convexity of the function $(1-\sqrt{x})^2$, the second inequality is obtained by Lemma~\ref{lem:ratio}, and $[x]^+ \triangleq \max\{x,0\}$.

Therefore, it can be obtained that
\begin{align}\label{sum}
&\frac{1}{K}\sum_{k=1}^{\imath(K)} \myexpectation{ \left(\frac{h_k}{h_{\imath(K)}} -1\right)^2} 
\!\!\geq \!
\frac{1}{K}\sum_{k=1}^{\imath(K)} \left(\left[1-\sqrt{\frac{K}{K-k+1} \frac{k+1}{\imath(K)-2}}\right]^+\right)^2\!\!,\forall \imath(K)\!>\!2\\
\label{change_k}
&\geq \frac{1}{K}  \sum_{k=2}^{\imath(K)} \left(\left[1-\sqrt{\frac{K}{K-k} \frac{k}{\imath(K)-2}}\right]^+\right)^2\\
\label{integral}
&\geq \frac{1}{\frac{\imath(K)-2}{K}} \sum_{k=2}^{g(K)}\! \left(\!\sqrt{\frac{\frac{g(K)}{K}}{1-\frac{g(K)}{K}}}\!-\! \sqrt{ \frac{\frac{k}{K}}{1-\frac{k}{K}}}\!\right)^2\!\frac{1}{K}\!=\!\frac{1}{\frac{\imath(K)-2}{K}} \left(\!\int_{0}^{\frac{g(K)}{K}}\left(\!\sqrt{\frac{\frac{g(K)}{K}}{1-\frac{g(K)}{K}}}-\sqrt{\frac{x}{1-x}}\!\right)^2\!\mathrm{d}x \!+\! o(\frac{1}{K})\!\!\right)\\ \label{last}
&=\frac{1}{\frac{\imath(K)-2}{K}} \left(\mu\left(\frac{g(K)}{K}\right)\right)+o(\frac{1}{\imath(K)}),
\end{align}
where
$
g(K)=\left\lfloor\frac{\imath(K)-2}{1+(\imath(K)-2)/K} \right\rfloor,
$
and $\mu(x)$ is defined in \eqref{haha} and it can be proved that
$
\mu(x)=\frac{1}{6}x^2+o(x^3), x\rightarrow 0.
$
\eqref{change_k} is due to $K-k+1>K-k-1$\footnote{Note that here we assume that $\imath(K)+1<K$ in \eqref{sum}. For the case that $\imath(K)+1 \geq K$, the summation in \eqref{sum} can be rewritten as $\sum_{k=1}^{\imath(K)-2} \left(\left[1-\sqrt{\frac{K}{K-k+1} \frac{k+1}{\imath(K)-2}}\right]^+\right)^2+\sum_{k=\imath(K)-1}^{\imath(K)} \left(\left[1-\sqrt{\frac{K}{K-k+1} \frac{k+1}{\imath(K)-2}}\right]^+\right)^2$. Following the similar steps of the $\imath(K)+1<K$ case, this one has the same asymptotic results as the $\imath(K)+1<K$ case and the detailed analysis is omitted for brevity.} and is obtained by replacing $(k+1)$ with $k$,
\eqref{integral} is due to the facts that $g(K)/K<1/2$ and the function $\sqrt{x/(1-x)}$ is monotonic and bounded in $[0,1/2]$, and is obtained by using Riemann integral to approximate Riemann sum when $K$ is large.

Assuming that $\limsup\limits_{K\rightarrow \infty}\frac{\imath(K)}{K} =c \neq 0$, from \eqref{last}, we have 
$
\limsup\limits_{K\rightarrow\infty} 
\frac{\myexpectation{\mse}}{K} 
\geq 
\limsup\limits_{K\rightarrow\infty}  \frac{1}{K} \times \\
\sum_{k=1}^{\imath(K)} \myexpectation{ \left(\frac{h_k}{h_{\imath(K)}} -1\right)^2} 
\geq \frac{1}{c} \mu (\frac{c}{1+c}),
$
completing the proof of \eqref{key_scaling1}.

Assuming that $\lim\limits_{K\rightarrow \infty}\frac{\imath(K)}{K}= 0$, we have
$
\mu(g(K)/K)= \frac{1}{6} \left(\frac{\imath(K)}{K}\right)^2+o\left(\left(\frac{\imath(K)}{K}\right)^3\right)+o(\frac{1}{\imath(K)}).
$
Taking it into \eqref{last} and jointly using \eqref{second} in \eqref{low_low}, it can be obtained that
\begin{equation} \label{lower}
\frac{\myexpectation{\mse}}{K} \geq \frac{1}{6} \frac{\imath(K)}{K} + \frac{\sigma^2}{P} \frac{1}{\imath(K)}+o\left(\frac{\imath(K)}{K}\right)+o(\frac{1}{\imath(K)}),\!K\!\rightarrow\!\infty.
\end{equation}
From the upper and lower bounds \eqref{inverse_gamma} and \eqref{lower}, \eqref{key_scaling} can be obtained.

\section*{Appendix C: Proof of Theorem~\ref{theory:power}}
From \eqref{Power}, an upper bound and a lower bound of the average power consumption can be obtained as
\begin{align} \label{upper2}
\frac{\myexpectation{\mathsf{PW}}}{K} & \leq \frac{P}{K}
\left(\imath(K)+ \sum_{k=\imath(K)+1}^{K} \myexpectation{\frac{h^2_{\imath(K)+1}}{h^2_{k}}}\right),\ 
\frac{\myexpectation{\mathsf{PW}}}{K}  \geq \frac{P}{K} \imath(K).
\end{align}

For the case that $\liminf_{K\rightarrow \infty} \imath(K)=K$, using the lower bound in \eqref{upper2} and the fact that $\pw/K\leq P$, we have $\limsup_{K\rightarrow \infty} \frac{\myexpectation{\mathsf{PW}}}{K} = P$, which completes the proof of \eqref{kkey_scaling2}.
For the case that $\limsup_{K\rightarrow \infty} \imath(K)/K=c'\neq 0$, using the lower bound in \eqref{upper2}, we have $\limsup_{K\rightarrow \infty} \frac{\myexpectation{\mathsf{PW}}}{K} \geq c'P$, which completes the proof of \eqref{kkey_scaling1}.

For the case that $\limsup_{K\rightarrow \infty} \imath(K)/K=0$, using \eqref{upper2}, we further have
\begin{align}\label{cauchy2}
&\frac{\myexpectation{\mathsf{PW}}}{K} 
 \!<\! \frac{P}{K}\!
\left(\!\imath(K)\!+ \sum_{k=\imath(K)+1}^{K} \frac{(\imath(K)+2)}{(K-\imath(K))}\frac{K}{(k-2)}\!\right), \forall i(K)\!<\!K\\ \label{harmonic}
& < \frac{P}{K}
\left(\!\imath(K)\!+\!\frac{(\imath(K)+2)K}{(K-\imath(K))} \sum_{k=1}^{K} \frac{1}{k}\!\right)
\! =\! P
\left(\!\frac{\imath(K)}{K}\!+\!\frac{\imath(K)+2}{K-\imath(K)} O(\log(K))\!\right),K\rightarrow\infty\\ \label{P_up}
&=O\left(\frac{\imath(K)\log(K)}{K}\right),K\rightarrow\infty,
\end{align}
and
$
\frac{\myexpectation{\mathsf{PW}}}{K}  \geq \frac{P}{K} \imath(K)= O\left(\frac{\imath(K)}{K}\right),K\rightarrow\infty,
$
where \eqref{cauchy2} is a consequence of Lemma~\ref{lem:ratio}, and the equality in \eqref{harmonic} is due to the property of the harmonic series, which completes the proof of \eqref{kkey_scaling}.


\vspace{-0.8cm}

\ifCLASSOPTIONcaptionsoff
\fi


\begin{thebibliography}{10}
	\providecommand{\url}[1]{#1}
	\csname url@samestyle\endcsname
	\providecommand{\newblock}{\relax}
	\providecommand{\bibinfo}[2]{#2}
	\providecommand{\BIBentrySTDinterwordspacing}{\spaceskip=0pt\relax}
	\providecommand{\BIBentryALTinterwordstretchfactor}{4}
	\providecommand{\BIBentryALTinterwordspacing}{\spaceskip=\fontdimen2\font plus
		\BIBentryALTinterwordstretchfactor\fontdimen3\font minus
		\fontdimen4\font\relax}
	\providecommand{\BIBforeignlanguage}[2]{{%
			\expandafter\ifx\csname l@#1\endcsname\relax
			\typeout{** WARNING: IEEEtran.bst: No hyphenation pattern has been}%
			\typeout{** loaded for the language `#1'. Using the pattern for}%
			\typeout{** the default language instead.}%
			\else
			\language=\csname l@#1\endcsname
			\fi
			#2}}
	\providecommand{\BIBdecl}{\relax}
	\BIBdecl
	
	\bibitem{wu2013data}
	X.~Wu, X.~Zhu, G.-Q. Wu, and W.~Ding, ``Data mining with big data,'' \emph{IEEE
		Trans. Knowl. Data Eng.}, vol.~26, no.~1, pp. 97--107, 2013.
	
	\bibitem{domo}
	\BIBentryALTinterwordspacing
	D.~Inc, \emph{Data Never Sleeps 6.0}, 2018. [Online]. Available:
	\url{https://www.domo.com/learn/data-never-sleeps-6}
	\BIBentrySTDinterwordspacing
	
	\bibitem{buck1979approximate}
	R.~C. Buck, ``Approximate complexity and functional representation,'' \emph{J.
		Math. Anal. Appl.}, vol.~70, pp. 280--298, 1979.
	
	\bibitem{GoldenbaumTCOM}
	M.~{Goldenbaum} and S.~{Stanczak}, ``Robust analog function computation via
	wireless multiple-access channels,'' \emph{IEEE Trans. Commun.}, vol.~61,
	no.~9, pp. 3863--3877, Sep. 2013.
	
	\bibitem{GoldenbaumTSP}
	M.~{Goldenbaum}, H.~{Boche}, and S.~{Stańczak}, ``Harnessing interference for
	analog function computation in wireless sensor networks,'' \emph{IEEE Trans.
		Signal Process.}, vol.~61, no.~20, pp. 4893--4906, Oct 2013.
	
	\bibitem{consensus}
	F.~Molinari, S.~Stanczak, and J.~Raisch, ``Exploiting the superposition
	property of wireless communication for average consensus problems in
	multi-agent systems,'' in \emph{Proc. European Control Conference (ECC)},
	2018, pp. 1766--1772.
	
	\bibitem{GuangxuWangyong}
	G.~Zhu, D.~Liu, Y.~Du, C.~You, J.~Zhang, and K.~Huang, ``Towards an intelligent
	edge: Wireless communication meets machine learning,'' \emph{arXiv preprint
		arXiv:1809.00343}, 2018.
	
	\bibitem{Gunduz}
	M.~M. Amiri and D.~Gunduz, ``Machine learning at the wireless edge: Distributed
	stochastic gradient descent over-the-air,'' \emph{arXiv preprint
		arXiv:1901.00844}, 2019.
	
	\bibitem{Osvaldo}
	J.-H. Ahn, O.~Simeone, and J.~Kang, ``Wireless federated distillation for
	distributed edge learning with heterogeneous data,'' \emph{arXiv preprint
		arXiv:1907.02745}, 2019.
	
	\bibitem{GoldenbaumWIOPT}
	M.~{Goldenbaum}, H.~{Boche}, and S.~{Stańczak}, ``Nomographic gossiping for
	f-consensus,'' in \emph{Proc. IEEE WiOpt}, May 2012, pp. 130--137.
	
	\bibitem{Katabi}
	\BIBentryALTinterwordspacing
	O.~Abari, H.~Rahul, and D.~Katabi, ``Over-the-air function computation in
	sensor networks,'' \emph{arXiv preprint}, 2016. [Online]. Available:
	\url{https://arxiv.org/pdf/1612.02307.pdf}
	\BIBentrySTDinterwordspacing
	
	\bibitem{GuangxuMIMO}
	G.~{Zhu} and K.~{Huang}, ``{MIMO} over-the-air computation for high-mobility
	multimodal sensing,'' \emph{IEEE Internet Things J.}, vol.~6, no.~4, pp.
	6089--6103, Aug 2019.
	
	\bibitem{Xiaoyang}
	X.~{Li}, G.~{Zhu}, Y.~{Gong}, and K.~{Huang}, ``Wirelessly powered data
	aggregation for {IoT} via over-the-air function computation: Beamforming and
	power control,'' \emph{IEEE Trans. Wireless Commun.}, vol.~18, no.~7, pp.
	3437--3452, Jul. 2019.
	
	\bibitem{Goldenbaumletter}
	M.~{Goldenbaum} and S.~{Stanczak}, ``On the channel estimation effort for
	analog computation over wireless multiple-access channels,'' \emph{IEEE
		Wireless Commun. Lett.}, vol.~3, no.~3, pp. 261--264, June 2014.
	
	\bibitem{CEO}
	T.~{Berger}, {Zhen Zhang}, and H.~{Viswanathan}, ``The {CEO} problem,''
	\emph{IEEE Trans. Inf. Theory}, vol.~42, no.~3, pp. 887--902, May 1996.
	
	\bibitem{Alex}
	C.~{Wang}, A.~S. {Leong}, and S.~{Dey}, ``Distortion outage minimization and
	diversity order analysis for coherent multiaccess,'' \emph{IEEE Trans. Signal
		Process.}, vol.~59, no.~12, pp. 6144--6159, Dec 2011.
	
	\bibitem{Cui}
	J.~{Xiao}, S.~{Cui}, Z.~{Luo}, and A.~J. {Goldsmith}, ``Linear coherent
	decentralized estimation,'' \emph{IEEE Trans. Signal Process.}, vol.~56,
	no.~2, pp. 757--770, Feb 2008.
	
	\bibitem{Jiang1}
	F.~{Jiang}, J.~{Chen}, and A.~L. {Swindlehurst}, ``Optimal power allocation for
	parameter tracking in a distributed amplify-and-forward sensor network,''
	\emph{IEEE Trans. Signal Process.}, vol.~62, no.~9, pp. 2200--2211, May 2014.
	
	\bibitem{AlexControl}
	A.~S. Leong, S.~Dey, G.~N. Nair, and P.~Sharma, ``Power allocation for outage
	minimization in state estimation over fading channels,'' \emph{IEEE Trans.
		Signal Process.}, vol.~59, no.~7, pp. 3382--3397, Jul. 2011.
	
	\bibitem{liu2019wireless}
	W.~Liu, P.~Popovski, Y.~Li, and B.~Vucetic, ``Wireless networked control
	systems with coding-free data transmission for industrial {IoT},'' \emph{IEEE
		Internet Things J.}, vol.~7, no.~3, pp. 1788--1801, 2020.
	
	\bibitem{Gastpar1}
	M.~{Gastpar}, ``Uncoded transmission is exactly optimal for a simple {Gaussian}
	``{Sensor}" network,'' \emph{IEEE Trans. Info. Theory}, vol.~54, no.~11, pp.
	5247--5251, Nov. 2008.
	
	\bibitem{CAO}
	\BIBentryALTinterwordspacing
	X.~Cao, G.~Zhu, J.~Xu, and K.~Huang, ``Optimal power control for over-the-air
	computation in fading channels,'' \emph{arXiv preprint}, 2019. [Online].
	Available: \url{https://arxiv.org/pdf/1906.06858.pdf}
	\BIBentrySTDinterwordspacing
	
	\bibitem{MACRate1}
	P.~Viswanath, D.~N.~C. Tse \emph{et~al.}, ``Sum capacity of the vector gaussian
	broadcast channel and uplink-downlink duality,'' \emph{IEEE Trans. Inf.
		Theory}, vol.~49, no.~8, pp. 1912--1921, 2003.
	
	\bibitem{MACMSE}
	E.~A. {Jorswieck}, B.~{Ottersten}, A.~{Sezgin}, and A.~{Paulraj}, ``Guaranteed
	performance region in fading orthogonal space-time coded broadcast
	channels,'' in \emph{Proc. IEEE ISIT}, 2007, pp. 96--100.
	
	\bibitem{GuangxuInference}
	G.~{Zhu}, S.~{Ko}, and K.~{Huang}, ``Inference from randomized transmissions by
	many backscatter sensors,'' \emph{IEEE Trans. Wireless Commun.}, vol.~17,
	no.~5, pp. 3111--3127, May 2018.
	
	\bibitem{orderstatistics}
	B.~C. Arnold, N.~Balakrishnan, and H.~N. Nagaraja, \emph{A first course in
		order statistics}.\hskip 1em plus 0.5em minus 0.4em\relax Siam, 1992,
	vol.~54.
	
\end{thebibliography}

\end{document}